\documentclass[review]{elsarticle}

\usepackage{amsmath,amsthm,amssymb,latexsym,stmaryrd}
\usepackage{hyperref} 
\usepackage{tikz}
\usepackage{tikz-cd}
\usepackage[mathscr]{euscript} 
\usepackage{stackrel}
\usepackage[all]{xy}

\journal{Journal of \LaTeX\ Templates}

\begin{document}

\begin{frontmatter}

\title{The $K_\infty$ Homotopy $\lambda$-Model \tnoteref{mytitlenote}\footnote{To appear in the Logic Journal of the IGPL, OUP.}}


\author{Daniel O. Martínez-Rivillas}

\author{Ruy J.G.B. de Queiroz}


\begin{abstract}
We extend the complete ordered set Dana Scott's $D_\infty$ to a complete weakly ordered Kan complex $K_\infty$, with properties that guarantee the non-equivalence of the interpretation of some higher conversions of $\beta\eta$-conversions of $\lambda$-terms.
\end{abstract}

\begin{keyword}
Homotopy domain theory\sep Homotopy domain equation\sep Homotopy lambda model \sep weakly ordered Kan complex  \sep complete homotopy partial order (c.h.p.o) \sep homotopy Scott domain
\MSC[2020] 03B70
\end{keyword}

\end{frontmatter}


\newtheorem{defin}{Definition}[section]
\newtheorem{teor}{Theorem}[section]
\newtheorem{corol}{Corollary}[section]
\newtheorem{prop}{Proposition}[section]
\newtheorem{rem}{Remark}[section]
\newtheorem{lem}{Lemma}[section]
\newtheorem{n}{Notation}[section]
\newtheorem{ejem}{Example}[section]

\section{Introduction}
\label{chap:capitulo1}

The initiative to search for $\lambda$-models \cite{DBLP:books/mk/HindleyS08} with an $\infty$-groupoid structure emerged in \cite{Martinez}, which studied the geometry of any complete partial order (c.p.o) (e.g., $D_\infty$), and found that the topology inherent in these models generated trivial higher groups. From that moment on, the need arose to look for a type of model that presented a rich geometric structure, where their higher-order fundamental groups would not collapse. 

\bigskip It is known in the literature that Dana Scott's Domain Theory (\cite{DBLP:books/mk/Abramsky94} and \cite{DBLP:books/mk/Asperti91}), provides general techniques for obtaining $\lambda$-models by solving domain equations over arbitrary Cartesian closed categories. To fulfill the purpose of getting $\lambda$-models with non-trivial $\infty$-groupoid structure, the most natural way would be to adapt Dana Scott's Domain Theory to a ``Homotopy Domain Theory'' (\cite{MartinezHoDT}, \cite{Martinez21}  and \cite{Tese}), where the Cartesian closed categories (c.c.c) will be replaced by Cartesian closed $\infty$-categories (c.c.i), in this work: the c.p.o's will be replaced by ``c.h.p.o's (complete homotopy partial orders)'', continuity with \cite{Martinez21}: the $0$-categories by $0$-$\infty$-categories and the isomorphism between objects in a Cartesian closed 0-category (in the recursive domain equation) will be replaced by an equivalence between objects in a Cartesian closed $0$-$\infty$-category, which we call ``homotopy domain equation'' such as summarized in Table \ref{Table HoDT}. Here, the last row we added to this table is: Extending the $D_\infty$ $\lambda$-model to the $K_\infty$ homotopy $\lambda$-model (initially defined in \cite{Tese}) so that proving its non-triviality and that is a ``Homotopy Scott Domain", are the new results of this work. 

	\begin{table}[ht]
		\caption{Generalization of Recursive Domain Equations} 
		\label{Table HoDT}
		\centering
		\begin{tabular}{p{6cm}p{5.5cm}}
			\hline
			\multicolumn{1}{c}{\textbf{Dana Scott's Domain Theory}} & \multicolumn{1}{c}{\textbf{Homotopy Domain Theory}}  \\
			\hline     
			Cartesian Closed Category (c.c.c)  & Cartesian Closed $\infty$-Category (c.c.i) 
			\\ 
			
			Complete Partial Order (c.p.o) &  Complete Homotopy Partial Order (c.h.p.o)           
			\\ 
            CPO c.c.c & CHPO c.c.i \\
			0-category & $0$-$\infty$-category \\
			Recursive Domain Equation: 
			
			$X\cong [X\rightarrow X]$	& Homotopy Domain Equation: 
			
			$X\simeq [X\rightarrow X]$ \\

            The $D_\infty$ $\lambda$-model  & The $K_\infty$ Homotopy $\lambda$-model 
			\\ \hline
		\end{tabular}
	\end{table}

\bigskip Intuitively, from the computational point of view, we have that a Kan complex, which satisfies the Homotopy Domain Equation as $K_\infty$, is not only capable of verifying the computability of constructions typical of classical programming languages, as $D_\infty$ does it, but also, it has the advantage (over $D_\infty$) of verifying the computability of higher constructions, such as a mathematical proof of some proposition, the proof of the equivalence between two proofs of the same proposition etc.

\bigskip To guarantee the existence of Kan complexes with good information, we define the ``non-trivial Kan complexes", which in intuitive terms are those that have holes in all higher dimensions and also have holes in the locality or class of any vertex. We prove that $K_\infty$ is a non-trivial Kan complex and is also reflexive, that is, $K_\infty$ is a ``non-trivial homotopy $\lambda$-model", previously called homotopic $\lambda$-model in \cite{Martinez}, and developed in \cite{MartinezHoDT}.  

\bigskip  In Section \ref{sec:infty-categories} this corresponds to the bases of $\infty$-categories and some result of previous work on homotopy domain theory. In Section \ref{Section: Complete Homotopy Partial Orders}, we extend the category of complete partial orders CPO to $\infty$-category of complete homotopy partial orders CHPO, we generalise Dana Scott's set $D_\infty$ to a Kan complex that we denote by $K_\infty$ and prove that it is a Scott Domain  extended to a CHPO. In Section \ref{Section: Contruction of K_infty}, we show that $K_\infty$ is a solution of the equation $X\simeq[X\rightarrow X]$, i.e., is a homotopy $\lambda$-model, in addition we prove that it is non-trivial, and finally, we present an example of the interpretation of $\beta$-contraction and $\eta$-contraction and the non-equivalence of these in $K_\infty$.

\section{Theoretical Foundations}
\label{sec:infty-categories}

The $\infty$-categories or quasi-categories are a type of simplicial sets, introduced by \cite{DBLP:books/mk/Boardman73} to generalise  classical categories. In the decade of the 2000s, André Joyal made important advances, e.g. showing some results analogous to classical category theory in its $\infty$-categorical version, as can be seen in \cite{Joyal2002} and \cite{Joyal2008}. In \cite{DBLP:books/mk/Lurie} and \cite{kerodon}, one can find the most complete treatise written so far on the theory of $\infty$-categories and in \cite{DBLP:books/mk/Lurie17} everything related to the higher operators and algebras on $\infty$-categories.  

\subsection{Simplicial sets}
\label{sub:simplicial-sets}

For a better understanding of the definitions and basic results of $\infty$-categories, necessary for the development of this work, we present some notions of simplicial sets \cite{DBLP:books/mk/Goerss09}.

\begin{defin}[Simplicial indexing category]
	Define $\Delta$ as the category as follows. The objects are finite ordinals $[n]=\{0,1,\ldots,n\}$,  $n\geq 0$, and  the morphisms are the (non-strictly) order-preserving maps. Morphisms in $\Delta$ are often called simplicial operators.
\end{defin}

\begin{rem}
 There is a coface operator $d^i: [n-1]\rightarrow [n]$, which skips the $i$-th element and a codegeneracy operator $s^i: [n+1]\rightarrow [n]$, which maps $i$ and $i+1$ to the same element. Every operator $f^{\ast}$ in $\Delta$ can be obtained as a finite composition of coface and codegeneracy operators.
\end{rem}

\begin{defin}[Simplicial set]
	A simplicial set $X$ is a functor $X : \Delta^{op}\rightarrow Set$ (or presheaf). A simplicial
	morphism is just a natural transformation of functors. The category of simplicial sets $Fun(\Delta^{op},Set)$ will be denoted by $sSet$ or $Set_\Delta$.
\end{defin}

It is typical to write $X_n$ for $X([n])$, and to call it the set of $n$-simplexes in $X$.

\begin{rem}
	Given a simplex $a\in X_n$ and a simplicial operator $f^{\ast}: [m]\rightarrow [n]$,  the function $f:X_n\rightarrow X_m$ is given by $f(a):=X(f^{\ast})(a)$. In this explicit language, a simplicial set consists of
	\begin{itemize}
		\item a sequence of sets $X_ 0 , X _1 , X_2,\ldots,$
	     \item functions $f: X_n\rightarrow X_m$ for each simplicial operator $f^{\ast}: [m]\rightarrow [n]$. 
	     \item $id(a)=a$ and $(gf)(a)=g(f(a))$ for all simplex $a$ and simplicial operators $f^{\ast}$ and $g^{\ast}$ whenever this
	     makes sense.
	\end{itemize}
For the coface operator $d^i:[n-1]\rightarrow [n]$, the face map is denoted by $d_i:X_n\rightarrow X_{n-1}$, $0\leq i\leq n$. For the codegeneracy operator $s^i:[n+1]\rightarrow [n]$, the degeneracy map is written by $s_i:X_n\rightarrow X_{n+1}$, $0\leq i\leq n$.
\end{rem} 

\begin{defin}[Product of simplicial sets \cite{Friedman2012}]\label{Def. Cartesian product of simpliciual sets}
	Let $X$ and $Y$ be simplicial sets. Their product $X\times Y$ is defined by
	\begin{enumerate}
		\item $(X\times Y)_n=X_n\times Y_n=\{(x,y)\,|\,x\in X_n,\,y\in Y_n \}$,
		\item if $(x,y)\in(X\times Y)_n$, then $d_i (x,y) = (d_i x, d_i y)$, 
		\item if $(x,y)\in(X\times Y)_n$, then $s_i (x,y) = (s_i x, s_i y)$. 
	\end{enumerate}
\end{defin}

\medskip Notice that there are evident projection maps $\pi_1 : X \times Y\rightarrow X$ and $\pi_2 : X\times Y\rightarrow Y$ given by $\pi_1 (x, y) = x$ and $\pi_2(x, y) = y$. These maps are clearly simplicial morphisms.

\begin{defin}[Standard $n$-simplex]
	The standard $n$-simplex $\Delta^n$ is the simplicial set defined by
	$$\Delta^n:=\Delta(-,[n]).$$
	That is, the standard $n$-simplex is exactly the functor represented by the object $[n]$.
\end{defin}

The standard 0-simplex $\Delta^0$ is the
terminal object in $sSet$; i.e., for every simplicial set $X$ there is a unique map $X\rightarrow\Delta^0$. Sometimes we
write $\ast$ instead of $\Delta^0$ for this object. The empty simplicial set $\emptyset$ is the functor $\Delta^{op}\rightarrow Set$ sending each $[n]$ to the empty set. It is the initial object in $sSet$, i.e., for every simplicial set $X$ there is a unique map $\emptyset\rightarrow X$. Besides, there is a bijection $sSet(\Delta^n,X)\cong X_n$; applying the Yoneda Lemma to category $\Delta$ \cite{DBLP:books/mk/Cisinski}.

\medskip A graphical representation of the convex hull of $\Delta^n$ is made up of the $n+1$ vertices $\langle 0\rangle,\langle 1\rangle,\ldots,\langle n\rangle$ and the faces are the injective simplicial operators, which are called non-degenerated cells.

\begin{defin}[\cite{DBLP:books/mk/Goerss09}]
 For two simplicial sets $X$, $Y$ we have a mapping  simplicial set, $Map(X,Y)$ defined as:
$$Map(X,Y)_n := sSet(X \times \Delta^n,Y).$$	
\end{defin}

\medskip Note, in particular, $Map(X, Y )_0 = sSet(X\times \Delta^0 ,Y)\cong sSet(X, Y)$ (bijection of sets). Sometimes, to simplify the notation, the simplicial set $Map(X,Y)$ will be written as $X^Y$ or $[X\rightarrow Y]$.

\medskip Next,  a collection of sub-objects of standard simplexes is defined, called “horns”.

\begin{defin}[Horns]
For each $n \geq 1$, there are subcomplexes $\Lambda_i^n\subset\Delta^n$  for each $0\leq i\leq n$. The horn  $\Lambda_i^n$ is the
subcomplex of $\Delta^n$ such that this is the largest sub-object that does
not include the face opposing the $i$-th vertex.

\medskip When $0<i<n$ one says that $\Lambda_i^n\subset\Delta^n$ 
is an inner horn. It is also said that it is a left horn if $i<n$ and a right horn if $0<i$.
\end{defin}	

For example, the horns inside $\Delta^1$ are just the vertices: the left horn, the right horn $\Lambda_0^1=\{0\}\subset\Delta^1$ and $\Lambda_1^1=\{1\}\subset\Delta^1$. Neither is an inner horn.  

\bigskip Other example. $\Delta^2$ has three horns: The left horn $\Lambda_0^2$, the internal horn $\Lambda_1^2$ and the right horn $\Lambda_1^2$.

\subsection{Definition of $\infty$-category and Kan complex}
\label{sub:Definiton-infty-categoy-Kan-complex}

\begin{defin}[$\infty$-category \cite{DBLP:books/mk/Lurie}]
	An $\infty$-category is a simplicial set $X$ which has the following property: for any $0<i<n$, any map $f_0:\Lambda_i^n\rightarrow X$ admits an
	extension $f:\Delta^n\rightarrow X$.	
\end{defin}

\begin{defin}
	From the above definition, we have the following special cases:
	\begin{itemize}
		\item $X$ is a Kan complex if there is an extension for each $0\leq i\leq n$.
		\item $X$ is a category if the extension exists uniquely \cite{DBLP:books/mk/Rezk22}.
		\item $X$ is a groupoid if the extension exists for all $0\leq i\leq n$ and is unique \cite{DBLP:books/mk/Rezk22}. 
	\end{itemize}
\end{defin}

In other words,  a Kan complex is an $\infty$-groupoid; composed of objects, 1-morphisms, 2-morphisms, etc.,  all invertible.

\begin{defin}
A functor of $\infty$-categories $X\rightarrow Y$ is exactly a morphism of simplicial sets. Thus, $Fun(X,Y)=Map(X,Y)$ must be a simplicial set of the functors from $X$ to $Y$. 
\end{defin}

\begin{n}
The notation $Map(X,Y)$  is usually used for simplicial sets, while $Fun(X,Y)$ is for $\infty$-categories. One will refer to morphisms in $Fun(X, Y)$ as natural transformations of functors,
and equivalences in $Fun(X, Y)$ as natural equivalences. 

\bigskip The composition from $n$-simplex $f:\Delta^n\rightarrow X$ (or $f\in X_n$) with a functor $F:X\rightarrow Y$, will be denoted as the image $F(f)\in Y_n$, where $n\geq 0$.

\bigskip A 1-simplex $f:\Delta^1\rightarrow X$, such that $f(0)=x$ and $f(1)=y$ will be denoted as a morphism $f:x\rightarrow y$ in the $\infty$-category $X$.

\bigskip An inner horn $\Lambda_1^2\rightarrow X$, which corresponds to composable morphisms $x\xrightarrow{f}y\xrightarrow{g}z$
in the $\infty$-category $X$, will be denoted by $(g,-,f)$ or in some cases to simplify notation it will be denoted by $g.f$.
\end{n}

\begin{prop}[\cite{DBLP:books/mk/Lurie} and \cite{DBLP:books/mk/Cisinski}]
	For every $\infty$-category $Y$, the simplicial set $Fun(X, Y)$ is an $\infty$-category.
\end{prop}

\bigskip The proof of the following theorem can be found in \cite{DBLP:books/mk/Lurie} and \cite{DBLP:books/mk/Cisinski}.

\begin{defin}[Trivial Fibration]
	A morphism $p:X\rightarrow S$ of simplicial sets is a trivial fibration if it has the right lifting property with respect to every inclusion $\partial\Delta^n\subseteq\Delta^n$, i.e., if given any diagram 
	\[\xymatrix{
	& {\Delta^n}\ar@{^{(}->}[dd]_{}\ar[rr]& &
	X \ar[dd]^{p}\ar@{<--}[lldd]\\ \\
	& \Delta^n\ar[rr] & & S
	& 
}\]
there exists a dotted arrow, as indicated, rendering the diagram commutative.
\end{defin}

\begin{prop}[\cite{kerodon}]\label{Kerodon: trivial fibration of infty-categories of functors}
    Let $p:X\rightarrow Y$ be a trivial fibration of simplicial sets. Then, for every simplicial set $B$, the induced map $Fun(B,X)\rightarrow Fun(B,Y)$ is a trivial fibration. 
\end{prop}

\begin{teor}[Joyal]
	A simplicial set $X$ is an $\infty$-category if and only if the canonical morphism 
	$$Fun(\Delta^2,X)\rightarrow Fun(\Lambda_1^2,X),$$
	is a trivial fibration. Thus, each fibre of this morphism is contractible.
\end{teor}

The above theorem guarantees the laws of coherence of the composition of 1-simplexes or morphisms of an $\infty$-category. In the sense that the composition of morphisms is unique up to homotopy, i.e., the composition is well-defined up to a space of choices being contractible (equivalent to $\Delta^0$).

\begin{defin}[Space of morphisms \cite{DBLP:books/mk/Rezk17}]
For two vertices $x,y$ in an $\infty$-category $X$, define the space of morphisms $X(x,y)$ by the following pullback diagram	
	\[\xymatrix{
	& {X(x,y)}\ar[d]_{}\ar[r]^{} &
	Fun(\Delta^1,X) \ar[d]^{(s,t)}\\
	& \Delta^0\ar[r]_{(x,y)} & X\times X
	& 
}\]
in the category $sSet$.
\end{defin}

\begin{prop}[\cite{DBLP:books/mk/Lurie} and \cite{DBLP:books/mk/Rezk17}]
	The morphism spaces $X(x,y)$ are Kan complexes.
\end{prop}

\subsection{Natural transformations and natural equivalence}

In category theory, we have the concept of isomorphism of objects. For the case of the $\infty$-categories, we will have the equivalence of objects (vertices) in the following sense.  

\begin{defin}[Equivalent vertices]
A morphism (1-simplex) $f:x\rightarrow y$ in an $\infty$-category $X$ is invertible (equivalence) if there is a morphism $g:y\rightarrow x$ in $X$, a pair of 2-simplexes $\alpha,\beta\in X_2$ such that 
$(g,-,f)\xrightarrow{\alpha}1_x$ and $(f,-,g)\xrightarrow{\beta}1_y$, i.e., we can fill in following diagram
 	\[\xymatrix{
 	& {x}\ar[dd]_{1_x}\ar[rr]^f& &
 	y \ar[dd]^{1_y}\ar@{-->}[lldd]^g\\ \\
 	& x\ar[rr]_{f}& & y
 	& 
 }\]
with the 2-simplexes $\alpha$ and $\beta$.	
\end{defin}

\begin{teor}[\cite{DBLP:books/mk/Lurie} and \cite{DBLP:books/mk/Cisinski}]
	Let $X$ be an $\infty$-category. The following are equivalent
	\begin{enumerate}
		\item Every morphism (1-simplex) in $X$ is an equivalence.
		\item  $X$ is a Kan complex.
	\end{enumerate}
\end{teor}

\subsection{Natural transformations and natural equivalence}

\begin{defin}[\cite{DBLP:books/mk/Cisinski} and \cite{DBLP:books/mk/Rezk17}]\label{Natural transformation}
	If $X$ is a simplicial set and $Y$ is an $\infty$-category, and if $F, G : X\rightarrow Y$ are two functors, a natural
	transformation from $F$ to $G$ is a map $H : X\times \Delta^1\rightarrow Y$ such that
	$$H(x,0)=F(x),\hspace{0.5cm}H(x,1)=G(x),$$
	for each vertex $x\in X$. Such a natural transformation is invertible or it is a natural equivalence if for any vertex $x\in X$, the
	induced morphism $F(x)\rightarrow G(x)$ (corresponding to the restriction of $H$ to
	$\Delta^1 \cong \{x\}\times\Delta^1$) is invertible in $Y$. If there is a natural equivalence from $F$ to $G$, we write $F\simeq G$.
\end{defin}

\begin{rem}
This means that for each vertex $x\in X$, one chooses a morphism $H_x:F(x)\rightarrow G(x)$ such that the following diagram 
 	\[\xymatrix{
	& {F(x)}\ar[dd]_{F(f)}\ar[rr]^{H_x}\ar@{-->}[rrdd]^g& &
	G(x) \ar[dd]^{G(f)}\\ \\
	& F(x')\ar[rr]_{H_{x'}}& & G(x')
	& 
}\]
commutes under some 2-simplexes $\alpha:g\rightarrow(G(f),-,H_x)$ and $\beta:(H_{x'},-,F(f))\rightarrow g$.
\end{rem}

\subsection{Categorical equivalences and homotopy equivalences}

\begin{defin}[Categorical equivalence \cite{DBLP:books/mk/Rezk17} and \cite{DBLP:books/mk/Lurie}]
A functor of $\infty$-categories $F:X\rightarrow Y$ is a categorical equivalence if there is another functor $G:Y\rightarrow X$, such that $GF\simeq 1_X$ and $FG\simeq 1_Y$.	
\end{defin}

\begin{rem}
	From the above definition, if  $F:X\rightarrow Y$ is a functor of Kan complexes, we say that $F$ is a homotopy equivalence.
\end{rem}

\begin{lem}[\cite{DBLP:books/mk/Rezk17} and \cite{DBLP:books/mk/Cisinski}]
	A functor of $\infty$-categories $F:X\rightarrow Y$ is a categorical equivalence
	if it satisfies the following two conditions:
	\begin{itemize}
		\item Fully Faithful (Embedding): For two objects $x,y\in X$ the induced functor of Kan complexes
		$$X(x,y)\rightarrow Y(Fx,Fy),$$
		is a homotopy equivalence.
		\item Essentially Surjective: For every object $y\in Y$ there exists an object $x\in Y$ such that $Fx$ is equivalent to $y$.
	\end{itemize}
\end{lem}

\subsection{The join of $\infty$-categories}
\label{sub:section3}

Next, the extension from join of categories to $\infty$-categories. This will make us able to define limit and colimit in an $\infty$-category.

\begin{defin}[Join \cite{DBLP:books/mk/Lurie}]
	Let $K$ and $L$ be simplicial sets. The join $K\star L$ is the simplicial set defined by
	$$(K\star L)_n:=K_n\cup L_n\cup\bigcup_{i+1+j=n}K_i\times L_j, \hspace{0.5cm}n\geq 0.$$ 
\end{defin}

\begin{ejem}[\cite{DBLP:books/mk/Groth15}]
	\begin{enumerate}
		\item If $K\in sSet$ and $L=\Delta^0$, then $K^\vartriangleright=K\star\Delta^0$ is the cocone or the right cone on $K$. Dually,  If $L\in sSet$ then $L^\vartriangleleft=\Delta^0\star L$ is the cone or the left cone on $L$.
		\item Let $K=\Lambda_0^2$. If we see this left horn as a pushout, the cocone $(\Lambda_0^2)^\vartriangleright$ is isomorphic to the square $\square=\Delta^1\times\Delta^1$, that is, to the filled in diagram
		 	\[\xymatrix{
			& {(0,0)}\ar[dd]_{}\ar[rr]^{}\ar[rrdd]& &
			(1,0) \ar[dd]^{}\\ \\
			& (0,1)\ar[rr]_{}& & (1,1)
			& 
		}\] 
	\end{enumerate}
\end{ejem}

\begin{prop}[\cite{DBLP:books/mk/Lurie}]
	\begin{enumerate}
		\item[(i)] For the standard simplexes one has an isomorphism $\Delta^i\star\Delta^j\cong \Delta^{i+1+j}$, $i,j\geq 0$, and these isomorphisms have obvious inclusions of $\Delta^i$ and $\Delta^j$.
		\item[(ii)] If $X$ and $Y$ are $\infty$-categories, then the join $X\star Y$ is an $\infty$-category.
	\end{enumerate}
\end{prop}

\subsection{The slice $\infty$-category}

In the case of the classical categories, if $A,B$ are categories and $p:A\rightarrow B$ is any functor, one can form the slice category $B_{/p}$ of the object over $p$ or cones on $p$. The following propositions allow us to define the slice $\infty$-category.   

\begin{prop}[\cite{Joyal2002}]
Let $K$ and $S$ be simplicial sets, and $p:K\rightarrow S$ be an arbitrary map. There is a simplicial set $S_{/p}$ such that there exists a natural bijection 
$$sSet(Y,S_{/p})\cong sSet_p(Y\star K,S),$$
where the subscript on right hand side indicates that we consider only those morphisms $f:Y\star K\rightarrow S$ such that $f|K=p$.
\end{prop}

\begin{prop}[Joyal]
	Let $X$ be an $\infty$-category and $K$ be a simplicial set. If $p:K\rightarrow X$ is a map of simplicial sets, then $X_{/p}$ is an $\infty$-category. Moreover, if $q:X\rightarrow Y$ is a categorical equivalence, then the induced map $X_{/p}\rightarrow Y_{/qp}$ is also a categorical equivalence.  
\end{prop}

\begin{defin}[Slice $\infty$-Categorical \cite{DBLP:books/mk/Lurie}]
		Let $X$ be an $\infty$-category, $K$  a simplicial set and $p:K\rightarrow X$  a map of simplicial sets. Define the slice $\infty$-category $X_{/p}$ of objects on $p$ or cones on $p$ (denoted by the diagram $\overline{p}:K^{\vartriangleright}\rightarrow X$, where $p = \overline{p}|_{K}$). Dually, $X_{p/}$ is the $\infty$-category of objects under $p$ or cocones on $p$ (denoted by $\overline{p}:K^{\vartriangleleft}\rightarrow X$).
\end{defin}

\begin{ejem}
	Let $X$ be an $\infty$-category and $x\in X$ be an object that corresponds to map  $x:\Delta^0\rightarrow X$. The objects of the $\infty$-category $X_{/x}$ of cones in $x$ are morphism $y\rightarrow x$ in $X$, and the morphism from $y\rightarrow x$ to $z\rightarrow x$ in $X_{/x}$, are the 2-simplexes 
			 	
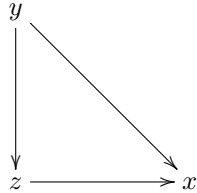
\[\xymatrix{
		& {y}\ar[dd]_{}\ar[rrdd]& &
		 \\ \\
		& z\ar[rr]_{}& & x
		& 
	}\]
in the $\infty$-category $X$. 
\end{ejem}

\subsection{Limits and colimits}

An object $t$ of a category is final if for each object $x$ in this category, there is a unique morphism $x\rightarrow t$. Next, we define the final objects in $\infty$-categories, under a contractible space of morphisms.

\begin{defin}[Final object \cite{DBLP:books/mk/Lurie}]
An object $\omega\in X$ in an $\infty$-category $X$, is a final object if for any object $x\in X$, the Kan complex of morphisms $X(x,\omega)$ is contractible.
\end{defin}

\begin{corol}[\cite{DBLP:books/mk/Cisinski}]\label{corolally-final-objects}
	The final objects of an $\infty$-category $X$ form a Kan complex which is either empty or equivalent to the point. 
\end{corol}

\begin{corol}[\cite{DBLP:books/mk/Cisinski}]
	Let $x$ be a final object in an $\infty$-category $X$. For any simplicial set $A$, the constant map $A\rightarrow X$ with value $x$ is a final object in $Map(A,X)$.  
\end{corol}

\begin{defin}(Limit and colimit \cite[\href{https://kerodon.net/tag/02JN}{Variant 02JN}]{kerodon})\label{Definition de Colimit and limit} 
	Let $X$ be an $\infty$-category and let $p:K\rightarrow X$ a map of simplicial sets. We will say that $\overline{p}:K^{\vartriangleleft}\rightarrow X$ is a colimit diagram of $p$ if it is an initial object of $X_{p/}$, and $\overline{p}:K^{\vartriangleright}\rightarrow X$  is a limit diagram of $p$ if it is a final object of $X_{/p}$. 
\end{defin}

By the dual of Corollary \ref{corolally-final-objects}, if the colimit exists, then the Kan complex of initial objects is contractible, i.e., the initial object is unique up to contractible choice.

\begin{defin}
 Let $K$ be a simplicial set, and $X$ be an $\infty$-category. We say that $X$ admits $K$-indexed colimits if for every $p:K\rightarrow X$ there exists its colimit. We define in the same way for the $K$-indexed limit in $X$. 
\end{defin}

\begin{prop}[\cite{kerodon}]\label{Colimits in Fun(S,X)}
    Let $X$ be an $\infty$-category which admits $K$-indexed colimits. Then, for every simplicial set $S$, the $\infty$-category $Fun(S,X)$ also admits K-indexed colimits. 
\end{prop}

\begin{prop}(\cite[\href{https://kerodon.net/tag/03D2}{Proposition 03D2}]{kerodon})\label{Proposition 03D2 kerodon} 
    Let $\operatorname{\mathcal{C}}$ be a category and let $\overline{\mathscr {F}}: \operatorname{\mathcal{C}}^{\triangleleft } \rightarrow \operatorname{\mathcal{D}}$ be a diagram of simplicial sets. The following conditions are equivalent: 
    \begin{enumerate}
        \item The diagram $\overline{\mathscr {F}}$ is a categorical colimit diagram.
        \item For every $\infty$-category $\operatorname{\mathcal{D}}$, the diagram of $\infty$-categories 
        $$( \operatorname{\mathcal{C}}^{\triangleright } )^{\operatorname{op}} \rightarrow \operatorname{QCat}\quad \quad C \mapsto \operatorname{Fun}( \overline{\mathscr {F}}(C), \operatorname{\mathcal{D}})$$
        is a categorical limit diagram \cite[\href{https://kerodon.net/tag/03BM}{Definition 03BM}]{kerodon}.
        \item For every $\infty$-category $\operatorname{\mathcal{D}}$, the diagram of Kan complexes
        $$( \operatorname{\mathcal{C}}^{\triangleright } )^{\operatorname{op}} \rightarrow \operatorname{Kan}\quad \quad C \mapsto \operatorname{Fun}( \overline{\mathscr {F}}(C), \operatorname{\mathcal{D}})^{\simeq }$$
        is a homotopy limit diagram \cite[\href{https://kerodon.net/tag/03B6}{Definition 03B6}]{kerodon}.
    \end{enumerate}
\end{prop}

Let $\operatorname{N}_{\bullet}$ be the nerve functor from the small categories to small $\infty$-categories, and $\operatorname{N}_{\bullet }^{\operatorname{hc}}$ be the homotopy coherent nerve functor from simplicial categories to the $\infty$-categories \cite{kerodon}. So we have the following results:

\begin{prop}(\cite[\href{https://kerodon.net/tag/03BB}{Corollary 03BB}]{kerodon})\label{Corollary 03BB kerodon}
    Let $\operatorname{\mathcal{C}}$ be a small category, let $\operatorname{\mathcal{D}}$ be a locally Kan simplicial category, and let $\overline{\mathscr {F}}: \operatorname{\mathcal{C}}^{\triangleleft } \rightarrow \operatorname{\mathcal{D}}$ be a functor. The following conditions are equivalent: 
    \begin{enumerate}
        \item The functor
        $$\operatorname{N}_{\bullet }^{\operatorname{hc}}(\overline{\mathscr {F}}): \operatorname{N}_{\bullet }(\operatorname{\mathcal{C}})^{\triangleleft } \simeq \operatorname{N}_{\bullet }(\operatorname{\mathcal{C}}^{\triangleleft }) \rightarrow \operatorname{N}_{\bullet }^{\operatorname{hc}}(\operatorname{\mathcal{D}})$$
        is a limit diagram in the $\infty$-category $\operatorname{N}_{\bullet }^{\operatorname{hc}}(\operatorname{\mathcal{D}})$ 
        \item  For every object $D \in \operatorname{\mathcal{D}}$, the functor
$$\operatorname{\mathcal{C}}^{\triangleleft } \rightarrow \operatorname{Kan}\quad \quad C \mapsto \operatorname{Hom}_{\operatorname{\mathcal{D}}}( D, \mathscr {F}(C) )_{\bullet }$$
is a homotopy limit diagram of Kan complexes. 
    \end{enumerate}
\end{prop}

\begin{prop}(\cite[\href{https://kerodon.net/tag/03DE}{Proposition 03DE}]{kerodon})\label{Proposition 03DE kerodon}
    Let  $\operatorname{\mathcal{C}}$ be a small filtered category and let $\overline{ \mathscr {F} }: \operatorname{\mathcal{C}}^{\triangleright } \rightarrow \operatorname{Set_{\Delta }}$ be a colimit diagram in the category of simplicial sets. Then $\overline{\mathscr {F}}$ is a categorical colimit diagram  \cite[\href{https://kerodon.net/tag/03D0}{Definition 03D0}]{kerodon}. 
\end{prop}

\subsection{Homotopy Domain Equation on an arbitrary Cartesian closed $\infty$-category}\label{HDE arbitrary cci}

This section is a direct generalization of the traditional methods for solving domain equations in Cartesian closed categories (see \cite{DBLP:books/mk/Asperti91} and \cite{DBLP:books/mk/Abramsky94}), in the sense of obtaining solutions for certain types of equations, which we call Homotopy Domain Equations in any Cartesian closed $\infty$-category.

\begin{defin}[\cite{DBLP:books/mk/Asperti91}]
	\begin{enumerate}
		\item An $\omega$-diagram in an $\infty$-category $\mathcal{K}$ is a diagram with the following structure:
		$$K_0 \stackrel{f_0}{\longrightarrow} K_1\stackrel{f_1}{\longrightarrow} K_2\longrightarrow\cdots \longrightarrow K_n \stackrel{f_n}{\longrightarrow}K_{n+1}\longrightarrow\cdots$$
		(dually, one defines $\omega^{op}$-diagrams by just reversing the arrows).
		\item An $\infty$-category $\mathcal{K}$ is $\omega$-complete ($\omega$-cocomplete) if it has limits (colimits) for all $\omega$-diagrams.
		\item A functor $F:\mathcal{K}\rightarrow\mathcal{K}$ is $\omega$-continuous if it preserves (under equivalence) all colimits of $\omega$-diagrams.  
	\end{enumerate}
\end{defin}

\begin{teor}[\cite{Martinez21}]\label{point-fixed-theorem}
	Let $\mathcal{K}$ be an $\infty$-category. Let $F:\mathcal{K}\rightarrow\mathcal{K}$ be a $\omega$-continuous (covariant) functor and take a vertex $K_0\in\mathcal{K}$ such that there is an edge  $\delta\in\mathcal{K}(K_0,FK_0)$. Assume also that $(K,\{\delta_{i,\omega}\in \mathcal{K}(F^iK_0,K)\}_{i\in\omega})$ is a colimit for the $\omega$-diagram $(\{F^iK_0\}_{i\in\omega},\{F^i\delta\}_{i\in\omega})$, where $F^0K_0=K_0$ and $F^0\delta=\delta$. Then $K\simeq FK$.
\end{teor}

\begin{defin}[\cite{Martinez21}]
An enriched $\infty$-category $\mathcal{K}$ on the h.p.o's, is a $0$-$\infty$-category if
	\begin{enumerate}
		\item every Kan complex $\mathcal{K}(A,B)$ is a c.h.p.o, whose minimum elements are noted by $0_{A,B}$, 
		\item composition of morphisms is a continuous operation with respect to the homotopy order,
		\item for every $f$ in $\mathcal{K}(A,B)$, $0_{B,C}.f\simeq 0_{A,C}$.
	\end{enumerate}
\end{defin}

\begin{defin}[h-projection pair \cite{Martinez21}]\label{h-projection}
	Let $\mathcal{K}$ be a $0$-$\infty$-category, and let $f^+:A\rightarrow B$ and $f^-:B\rightarrow A$ be two morphisms in $\mathcal{K}$. Then $(f^+,f^ -)$ is a homotopy projection (or h-projection) pair (from $A$ to $B$) if $f^-.f^+\simeq I_A$ and $f^+.f^-\precsim I_B$. If $(f^+,f^-)$ is an h-projection pair, $f^+\in\mathcal{K}^{HE}(A,B)$ is an h-embedding (homotopy embedding) and $f^-\in\mathcal{K}^{HP}(A,B)$ is an h-projection (homotopy projection). Where $\mathcal{K}^{HE}$ is the subcategory of $\mathcal{K}$ with the same objects and the h-embeddings as morphisms, and $\mathcal{K}^{HP}$ is the subcategory of $\mathcal{K}$ with the same objects and the h-projections as morphisms.               
\end{defin}

\begin{defin}[h-projections pair $0$-$\infty$-category \cite{Martinez21}]\label{h-projection-inftycategory}
	Let $\mathcal{K}$ be a $0$-$\infty$-category. The  $0$-$\infty$-category $\mathcal{K}^{HPrj}$ is the $\infty$-category embedding in $\mathcal{K}^{HE}$ with the same objects of $\mathcal{K}$ and h-projection pairs $(f^+,f^-)$ as morphisms.
\end{defin}

\begin{rem}
	Every h-embedding $i$ has unique (under homotopy) associated h-projection  $j= i^R$ (and, conversely, every h-projection $j$ has a unique (under homotopy) associated h-embedding $i=j^L$), since if there is $j_0$ such that $j_0.i\simeq I$ and $i.j_0\precsim I$ (under homotopy), so $j_0\succsim j_0.i.j\simeq j$ and $j_0\precsim j$ (under homotopy). Thus, $j_0\simeq j$ under homotopy (and, in the same way, we have $i_0\simeq i$ under homotopy). $\mathcal{K}^{HPrj}$ and  is equivalent to a subcategory  $\mathcal{K}^{HE}$ of $\mathcal{K}$ that has h-embeddings as morphisms  (as well to a subcategory $\mathcal{K}^{HP}$ of $\mathcal{K}$ which has h-projections as morphisms).  
\end{rem}

\begin{defin}[\cite{Martinez21}]
	Given a $0$-$\infty$-category $\mathcal{K}$, and a contravariant functor in the first component $F:\mathcal{K}^{op}\times\mathcal{K}\rightarrow\mathcal{K}$, the functor covariant $F^{+-}:\mathcal{K}^{HPrj}\times\mathcal{K}^{HPrj}\rightarrow\mathcal{K}^{HPrj}$ is defined by
	\begin{align*}
		&F^{+-}(A,B)=F(A,B), \\
		&F^{+-}((f^+,f^-),(g^+,g^-))=(F(f^-,g^+),F(f^+,g^-)),
	\end{align*}
	where $A,B$ are vertices and $(f^+,f^-),(g^+,g^-)$ are $n$-simplexes pairs in $\mathcal{K}^{HPrj}$.
\end{defin}

\begin{teor}[\cite{Martinez21}]\label{colimit-HPrj-category}
	Let $\mathcal{K}$ be a  $0$-$\infty$-category. Let $(\{K_i\}_{i\in\omega},\{f_i\}_{i\in\omega})$ be an $\omega$-diagram in $\mathcal{K}^{HPrj}$. If $(K,\{\gamma_i\}_{i\in\omega})$ is a limit for $(\{K_i\}_{i\in\omega},\{f_i^-\}_{i\in\omega})$ in $\mathcal{K}$, then $(K,\{(\delta_i,\gamma_i)\}_{i\in\omega})$ is a colimit for $(\{K_i\}_{i\in\omega},\{f_i\}_{i\in\omega})$ in $\mathcal{K}^{HPrj}$ (that is, every $\gamma_i$ is a right member of an h-projection pair). 
\end{teor}

Therefore, the following corollary is an immediate consequence of the previous theorem.

\begin{corol}[\cite{Martinez21}]\label{Corollary-universal-colimit}
	The cocone $(K,\{(\delta_i,\gamma_i)\}_{i\in\omega})$ for the $\omega$-chain $(\{K_i\}_{i\in\omega},\{(f^+_i,f^-_i)\}_{i\in\omega})$ in $\mathcal{K}^{HPrj}$ is universal (a cocone colimit) iff $\Theta=\bigcurlyvee_{i\in\omega}\delta_i.\gamma_i\simeq I_K$.
\end{corol} 

\begin{defin}[\cite{Martinez21} Locally monotonic]
	Let $\mathcal{K}$ be a $0$-$\infty$-category. A functor $F:\mathcal{K}^{op}$$\times\mathcal{K}\rightarrow\mathcal{K}$	is locally h-monotonic if it is monotonic on the Kan complexes of 1-simplexes, i.e., for $f,f'\in \mathcal{K}^{op}(A,B)$ and $g,g'\in \mathcal{K}(C,D)$ one has 
	$$f\precsim f' \, , g\precsim g'\Longrightarrow F(f,g)\precsim F(f',g').$$ 
\end{defin}

\begin{prop}[\cite{Martinez21}]
	If $F:\mathcal{K}^{op}$$\times\mathcal{K}\rightarrow\mathcal{K}$ is locally h-monotonic and $(f^+,f^-)$, $(g^+,g^-)$ are h-projection pairs, then $F^{+-}((f^+,f^-),(g^+,g^-))$ is also an h-projection pair.
\end{prop}

\begin{defin}[Locally continuous \cite{Martinez21}]
	Let $\mathcal{K}$ be a $0$-$\infty$-category. A $F:\mathcal{K}^{op}\times\mathcal{K}\rightarrow\mathcal{K}$ is locally continuous if it is $\omega$-continuous on the Kan complexes of 1-simplexes. That is, for every directed diagram $\{f_i\}_{i\in\omega}$ in $\mathcal{K}^{op}(A,B)$, and every directed diagram $\{g_i\}_{i\in\omega}$ in $\mathcal{K}(C,D)$, one has
	\begin{center}
		$F(\bigcurlyvee_{i\in\omega}\{f_i\},\bigcurlyvee_{i\in\omega}\{g_i\})\simeq \bigcurlyvee_{i\in\omega}F(f_i,g_i).$ 
	\end{center} 
\end{defin}

\begin{rem}
	If $F$ is locally continuous, then it is also locally monotonic.
\end{rem}

\begin{teor}[\cite{Martinez21}]\label{continuous-functor}
	Let $\mathcal{K}$ be a $0$-$\infty$-category. Let also $F:\mathcal{K}^{op}\times\mathcal{K}\rightarrow\mathcal{K}$ be a locally continuous functor. Then the functor $F^{+-}:\mathcal{K}^{HPrj}\times\mathcal{K}^{HPrj}\rightarrow\mathcal{K}^{HPrj}$ is $\omega$-continuous. 
\end{teor}

\begin{rem}[\cite{Martinez21}]\label{Remark-Solution-HDE}
	Let $\mathcal{K}$ be a Cartesian closed $0$-$\infty$-category, $\omega^{op}$-complete and with final object. Since the exponential functor  $\Rightarrow:\mathcal{K}^{op}\times\mathcal{K}\rightarrow\mathcal{K}$ and the diagonal functor $\Delta:\mathcal{K}\rightarrow\mathcal{K}\times\mathcal{K}$
	are locally continuous, by the Theorem \ref{continuous-functor}, the associated functors $$(\Rightarrow)^{+-}:\mathcal{K}^{HPrj}\times\mathcal{K}^{HPrj}\rightarrow\mathcal{K}^{HPrj}, \hspace{0.5cm} (\Delta)^{+-}:\mathcal{K}^{HPrj}\rightarrow\mathcal{K}^{HPrj}\times\mathcal{K}^{HPrj}$$
	are $\omega$-continuous. But the composition of $\omega$-continuous functors is still an $\omega$-continuous functor. Thus, the functor 
	$$F=(\Rightarrow)^{+-}.(\Delta)^{+-}:\mathcal{K}^{HPrj}\rightarrow\mathcal{K}^{HPrj},$$
	is $\omega$-continuous. By Theorem \ref{point-fixed-theorem} the functor $F$ has a fixed point, that is, there is a vertex $K\in\mathcal{K}$ such that $K\simeq(K\Rightarrow K)$. The $\infty$-category of the fixed points of $F$ is denoted by $Fix(F)$.
\end{rem}

\section{The $\infty$-category $CHPO$}
\label{Section: Complete Homotopy Partial Orders}
In this section, we reintroduce the complete homotopy partial orders (c.h.p.o) as a direct generalisation of the c.p.o's and $\infty$-category of the c.h.p.o's (noted by $CHPO$), where the sets are replaced by Kan complexes and the order relations $\leq$ by weak order relations  $\precsim$ (\cite{MartinezHoDT} and \cite{Martinez2HoDTvsHoTT21}). We define $K_\infty$ c.h.p.o from an analytical perspective and show that it is a Scott domain extended to weakly ordered Kan complexes.

\begin{defin}[\cite{kerodon}]
    A Kan complex is (-1)-truncated if is contractible or empty.
\end{defin}

\begin{defin}[\cite{kerodon}]
    An $\infty$-category $\mathcal{C}$ is a 0-category if the Kan complexes $\mathcal{C}(x,y)$ are (-1)-truncated.
\end{defin}

\begin{defin}
    A Kan fibration $p: X \rightarrow Y$ of simplicial sets is (-1)-truncated if its fibres are (-1)-truncated.\footnote{Specially, if $X$, $Y$ are $\infty$-categories, the functor $p$ is $1$-faithful (see \cite[\href{https://kerodon.net/tag/05AX}{Tag 05AX}]{kerodon} for the case $n=1$), that is, $p$ is faithful (\cite[\href{https://kerodon.net/tag/05AY}{Tag 05AY}]{kerodon}).}
\end{defin}

\begin{defin}[h.p.o]\label{Definition h.p.o}
    Let $\hat{K}$ be an $\infty$-category, the largest Kan complex $K\subseteq\hat{K}$ is a homotopy partial order (h.p.o) if the map 
    \[Fun(\Delta^1,\hat{K})\xrightarrow{(s,t)} \hat{K}\times \hat{K}\] is a (-1)-truncated fibration (in particular, if is a faithful functor).
\end{defin}

\begin{rem}\label{h.p.o}
	In other words, if $\hat{K}$ is an $\infty$-category. We say that the largest Kan complex $K\subseteq\hat{K}$ is a homotopy partial order (h.p.o), if for every $x,y\in K$ one has $\hat{K}(x,y)$ it is contractible or empty (i.e., if $\hat{K}$ is a 0-category). Since the pullback diagram 
 \[\xymatrix{
	& {\hat{K}(x,y)}\ar@{-->}[d]_{}\ar@{-->}[r]^{} &
	Fun(\Delta^1,\hat{K}) \ar[d]^{(s,t)}\\
	& \Delta^0\ar[r]_{(x,y)} & \hat{K}\times \hat{K}
	& 
}\]
preserves (-1)-truncated fibrations in $sSet$.\footnote{Similarly, preserves functors 1-faithful (\cite[\href{https://kerodon.net/tag/05BD}{Tag 05BD}]{kerodon} for $n = 1$).} Thus, the space $\hat{K}(x,y)$ is contractible or empty.

\medskip The Kan complex $K\subseteq\hat{K}$ admits a relation of h.p.o $\precsim$ defined for each $x,y\in K$ as follows:
$$x\precsim y, ~\text{if}~ \hat{K}(x,y)\neq\emptyset,$$
then, the pair $(K,\precsim)$ is a h.p.o. (we denote only $K$, which, of course, has a $\hat{K}$ $\infty$-category associated with it, as in \cite{MartinezHoDT} and \cite{Tese}). 
\end{rem}

\begin{defin}\label{suborder of homotopy}
    Let $K$ be a h.p.o. We say that a h.p.o $X\subseteq K$, if  $\hat{X}\subseteq\hat{K}$ is a full subcategory, i.e., $X$ inherits the weak order of $K$ and all its higher simplexes. 
\end{defin}

\begin{defin}[c.h.p.o. \cite{MartinezHoDT} and \cite{Tese}] 
	Let $K\subseteq\hat{K}$ be h.p.o.
	\begin{enumerate}
		\item A h.p.o $X\subseteq K$ is directed if $X\neq \emptyset$ and for each $x,y\in X$, there exists $z\in X$ such that $\hat{K}(x, z)$ and $\hat{K}(y, z)$ are not empty (i.e., $x\precsim z$ and $y\precsim z$).
		\item $K$ is a complete homotopy partial order (c.h.p.o) if
		\begin{enumerate}
			\item There are minimum elements, i.e.,  $0\in K$ is a  minimum element if for each $x\in K$, $\hat{K}(0,x)$ is not empty. That is, $0$ is an initial object of the 0-category $\hat{K}$.
			\item For each directed $X\subseteq K$ the supremum (or colimit) $\bigcurlyvee X\in K$ exists. 
		\end{enumerate}
	\end{enumerate}
\end{defin}

\begin{rem}\label{Remark}
	If $K$ is a c.h.p.o., note that $\hat{K}$ is weakly contractible (i.e., its geometric realisation $|\hat{K}|$ is contractible) but $K$ is not necessarily contractible: for example, if $\hat{K} = \langle Sing(|\partial\Delta^n|)\cup \{0\}\rangle$ is the smallest 0-category that contains the Kan complex $K=Sing(|\partial\Delta^n|)\cup \{0\}$ such that $0$ be the initial object unique of $\hat{K}$, where $Sing: Top \rightarrow Kan$  be the singular functor from the category of  topological spaces to the category of Kan complexes and $|\,\,|:Kan \rightarrow Top$ be the geometric realisation. Since $Sing(|\partial\Delta^n|)$ is isomorphic to the sphere $S^{n-1}$, we denote $Sing(|\partial\Delta^n|)$ by $S^{n-1}$. 
\end{rem}

\begin{defin}[Continuity \cite{MartinezHoDT} and \cite{Martinez21}]
	Let $K$ and $K'$  be a c.h.p.o. A functor $f:K\rightarrow K'$ is continuous if $f(\bigcurlyvee X)\simeq\bigcurlyvee f(X)$, where $f(X)$ is the essential image to be directed $X\subseteq K$.  
\end{defin}

\begin{prop}[\cite{Tese}]
	Continuous functors on c.h.p.o's are always monotonic.
\end{prop}
\begin{proof}[Proof]
	Let $f:K\rightarrow K'$ be a continuous functor between c.h.p.o's and suppose the non-trivial case $a\precsim b$ in $K$. Since the h.p.o $\{a,b\}\subseteq K$ is directed, by the continuity of $f$ we have $$f(a)\precsim \bigcurlyvee \{f(a),f(b)\}\simeq f(\bigcurlyvee \{a,b\})=f(b).$$
\end{proof}

\begin{prop}
If $K\subseteq\hat{K}$ and $K'\subseteq\hat{K'}$ are h.p.o's. Then, the Cartesian product $K\times K'\subseteq \hat{K}\times\hat{K'}$ is also a h.p.o.
\end{prop}
\begin{proof}
    By Definition \ref{Definition h.p.o}, we will prove that the map
    \[Fun(\Delta^1,\hat{K}\times\hat{K'})\xrightarrow{(s\times s', \, t\times t')} (\hat{K}\times \hat{K'})\times(\hat{K}\times \hat{K'})\]
is a (-1)-truncated fibration. Note that we can identify $(s\times s', \, t\times t')$ with the canonical map 
    \[Fun(\Delta^1,\hat{K})\times Fun(\Delta^1,\hat{K'}) \xrightarrow{(s,t)\times(s',t')} (\hat{K}\times \hat{K})\times(\hat{K'}\times \hat{K'}),\]
    since $(s,t)$ and $(s',t')$ are (-1)-truncated fibrations, $(s,t)\times(s',t')$ is also (-1)-truncated. 
\end{proof}

\begin{rem}\label{hpo KxK'}
    Replacing $\hat{K}$ with $\hat{K}\times \hat{K}'$ in the pullback of the Remark \ref{h.p.o}, where $K\subseteq\hat{K}$ and $K'\subseteq\hat{K'}$ are h.p.o's, we find that the space $(\hat{K}\times\hat{K'})((x,x'),(y,y'))$ is contractible or empty. Thus, the product of Kan complexes $K'\times K\subseteq\hat{K'}\times\hat{K}$ admits a relation of h.p.o $\precsim$, defined for each pair $(x,x'),(y,y')\in K\times K'$ as follows:
	$$(x,x')\precsim (y,y'),~\text{if}~(\hat{K}\times\hat{K'})((x,x'),(y,y'))\neq\emptyset,~\text{if}~\hat{K}(x,y)\times\hat{K'}(x',y')\neq\emptyset$$
\end{rem}

The Cartesian product between c.h.p.o's can be considered again as a c.h.p.o.

\begin{prop}[\cite{Tese}]
	Given the c.h.p.o's $K,K'$, let $K\times K'$ the Cartesian product partially ordered, according to Remark \ref{hpo KxK'}, by
	$$(x,x')\precsim (y,y'),~\text{if}~x\precsim y~\text{and}~x'\precsim' y'.$$
	Then $K\times K'$ is a c.h.p.o with for directed $X\subseteq K\times K'$
	$$\bigcurlyvee X=(\bigcurlyvee X_0,\bigcurlyvee X_1)$$
	where $X_0$ is the projection of $X$ on $K$, and $X_1$ is the projection of $X$ on $K'$. 
\end{prop}
\begin{proof}[Proof]
	$(0,0')$ is a minimum element of $K\times K'$. On the other hand, if $X\subseteq K\times K'$ is directed, by the definition of order in $K\times K'$, the projections $X_0$ and $X_1$ are also directed, so the supremum $\bigcurlyvee X\in K\times K'$ exists.
\end{proof}

\begin{prop}
    If $K\subseteq\hat{K}$ is an h.p.o. Then, for every simplicial set $S$, the $\infty$-category $Fun(S,\hat{K})$ is also an h.p.o.
\end{prop}
\begin{proof}
    According to Definition \ref{Definition h.p.o}, we see that the map 
    \[Fun(\Delta^1,Fun(S,\hat{K}))\xrightarrow{(s,t)} Fun(S,\hat{K})\times Fun(S,\hat{K})\]
    is a (-1)-truncated fibration. The map $(s,t)$ can be identified with the canonical map
\[Fun(S,Fun(\Delta^1,\hat{K}))\xrightarrow{} Fun(S,\hat{K}\times \hat{K}),\]
which is a (-1)-truncated fibration (by Definition \ref{Definition h.p.o}, \cite[\href{https://kerodon.net/tag/05AX}{Tag 05AX}]{kerodon} and Corollary \cite[\href{https://kerodon.net/tag/056S}{Tag 056S}]{kerodon} for $n = 1$).
\end{proof}

\begin{rem}\label{Order in Fun(S,X)}
    Changing $\hat{K}$ with $Fun(S,\hat{K})$ in the pullback of Remark \ref{h.p.o}, where $K\subseteq\hat{K}$ is h.p.o, we see that the space $(Fun(S,\hat{K}))(f,g)$ is contractible or empty. So, the Kan complex $Fun(S,K)\subseteq Fun(S,\hat{K})$ admits a relation $\precsim$, defined for each $f,g\in Fun(S,K)$ as follows:
    \begin{align*}
        f\precsim g,~&\text{if}~(Fun(S,\hat{K}))(f,g)\neq\emptyset, \\
                    &\text{if there is a 'unique' map}~f\xrightarrow{} g~\text{in}~Fun(S,\hat{K}), \\
                    &\text{if there is a 'unique' map}~f(x)\xrightarrow{} g(x)~\text{in}~\hat{K},~\text{for each}~x\in S,\\
                    &\text{if}~f(x)\precsim g(x),~\text{for each}~x\in S.
    \end{align*}
    Recall that the existence of a 'unique' map is up to homotopy, that is, when the space of choices, which in this case corresponds to $(Fun(S,\hat{K}))(f,g)$ and $\hat{K}(f(x),g(x))$, is contractible.
    \end{rem}

\begin{defin}[\cite{Tese}]\label{The chpo [K rightarrow K']}
	Let $K,K'$ be c.h.p.o's. Define the full subcategory $[K\rightarrow K']\subseteq Fun(K,K')$ of the continuous functors. By Remark \ref{Order in Fun(S,X)}, we can define the order pointwise on $[K\rightarrow K']$ by:
	$$f\precsim g\,\Longleftrightarrow \, \forall x\in K,\, f(x)\precsim' g(x).$$
\end{defin}

\begin{n}
	Let $K$ be an h.p.o and $P$ a predicate. Denote the full subcategory $\langle x\in K\,|\,P(x)\rangle\subseteq K$ as the h.p.o induced by the order of $K$ (i.e., $\langle x\in \hat{K}\,|\,P(x)\rangle\subseteq \hat{K}$ is a full subcategory), whose objects are the $x\in K$ that satisfy the property $P$ and the higher simplexes are the same as $K$.
\end{n}

\begin{lem}[\cite{Tese}]\label{lemma-complete}
	Let $\langle f_i\rangle_i\subseteq [K\rightarrow K']$ be an indexed directed of functors. Define 
	$$f(x)=\bigcurlyvee_i f_i(x).$$
	Then $f$ is well defined and continuous.
\end{lem}
\begin{proof}[Proof]
	Since $\langle f_i\rangle_i$ is directed, $\langle f_i(x)\rangle_i$ is directed for each $x\in K$, the functor $f$ exists. On the other hand, for directed $X\subseteq K$
	$$f(\bigcurlyvee X)=\bigcurlyvee_if_i(\bigcurlyvee X)\simeq\bigcurlyvee_i\bigcurlyvee f_i(X)\simeq\bigcurlyvee\bigcurlyvee_i f_i(X)=\bigcurlyvee f(X).$$
\end{proof}

\begin{prop}[\cite{Tese}]\label{Proposition-Continuous-Funtor-chpo}
	$[K\rightarrow K']$ is a c.h.p.o with supremum of a directed $F\subseteq [K\rightarrow K']$ defined by
	$$(\bigcurlyvee F)(x)=\bigcurlyvee\langle f(x)\,|\,f\in F\rangle.$$
\end{prop}
\begin{proof}[Proof]
	The constant functor $\boldsymbol{\lambda} x.0'$ is a minimum element of $[K\rightarrow K']$ . By Lemma \ref{lemma-complete} the functor $\boldsymbol{\lambda} x.\bigcurlyvee\langle f(x)\,|\,f\in F\rangle$ is continuous, which is the supremum of $F$.
\end{proof}

\begin{lem}[\cite{Tese}]\label{Lemma-bifunctor-continuous}
	Let $f:K\times K'\rightarrow K''$. Then $f$ is continuous iff $f$ is continuous in its arguments separately, that is, iff $\boldsymbol{\lambda} x.f(x,x_0')$ and $\boldsymbol{\lambda} x'.f(x_0,x')$ are continuous for all $x_0,x_0'$.
\end{lem}
\begin{proof}[Proof]
	$(\Rightarrow)$ Let $g=\boldsymbol{\lambda} x.f(x,x_0')$. Then for the directed $X\subseteq K$
	\begin{align*}
		g(\bigcurlyvee X)&=f(\bigcurlyvee X,x'_0) \\
		&\simeq\bigcurlyvee f(X\times  \{x_0\});\hspace{0.5cm}f~\text {is continuous and}~X\times\{x_0\}~\text{is directed} \\
		&=\bigcurlyvee g(X).
	\end{align*}
\end{proof}

Similarly $\boldsymbol{\lambda} x'.f(x_0,x')$ is continuous.

\bigskip $(\Leftarrow)$ Let $X\subseteq K\times K'$ be directed. So
\begin{align*}
	f(\bigcurlyvee X)&=f(\bigcurlyvee X_0,\bigcurlyvee X_1) \\
	&\simeq\bigcurlyvee f(X_0,\bigcurlyvee X_1);\hspace{0.5cm}~\text {by hypothesis,}\\
	&\simeq\bigcurlyvee \bigcurlyvee f(X_0, X_1);\hspace{0.5cm}~\text {by hypothesis,}\\
	&=\bigcurlyvee f(X);\hspace{0.5cm}~\text {$X$ is directed.}\\
\end{align*}

\begin{prop}[Continuity of application \cite{Tese}]\label{Proposition-Continuity-application}
	Define application 
	$$Ap:[K\rightarrow K']\times K\rightarrow K'$$
	by the functor $Ap(f,x)=f(x)$. Then $Ap$ is continuous. 
\end{prop}
\begin{proof}[Proof]
	The functor $\boldsymbol{\lambda} x.f(x)=f(x)$ is continuous by the continuity of $f$. Let $h=\boldsymbol{\lambda}f.f(x)$. Then for directed $F\subseteq [K\rightarrow K']$
	\begin{align*}
		h(\bigcurlyvee F)&=(\bigcurlyvee F)(x) \\
		&=\bigcurlyvee \langle f(x)\,|\,f\in F\rangle\hspace{0.5cm}\text{by Proposition \ref{Proposition-Continuous-Funtor-chpo}}, \\
		&=\bigcurlyvee \langle h(f)\,|\,f\in F\rangle \\
		&=\bigcurlyvee h(F).
	\end{align*}
	So $h$ is continuous and according to Lemma \ref{Lemma-bifunctor-continuous} the functor $Ap$ is continuous.
\end{proof}

\begin{prop}[Continuity of abstraction \cite{Tese}]\label{Proposotion-Continuity-abstraction}
	Let $f\in [K\times K'\rightarrow K'']$. Define the functor $\hat{f}(x)=\boldsymbol{\lambda} y\in K'.f(x,y)$. Then 
	\begin{enumerate}
		\item $\hat{f}$ is continuous, i.e., $\hat{f}\in [K\rightarrow [K'\rightarrow K'']]$;
		\item $\boldsymbol{\lambda}f.\hat{f}:[K\times K'\rightarrow K'']\rightarrow[K\rightarrow [K'\rightarrow K'']]$ is continuous.
	\end{enumerate}
\end{prop}
\begin{proof}[Proof]
	(1) Let $X\subseteq K$ be directed. Then
	\begin{align*}
		\hat{f}(\bigcurlyvee X)&=\boldsymbol{\lambda}y.f(\bigcurlyvee X,y) \\
		&\simeq\boldsymbol{\lambda}y.\bigcurlyvee f( X,y) \\
		&\simeq\bigcurlyvee \boldsymbol{\lambda}y.f( X,y);\hspace{0.3cm}\text{by Proposition \ref{Proposition-Continuous-Funtor-chpo}: takes $F=\boldsymbol{\lambda}y.f(X,y)$,}
	\end{align*}  
	(2) Let $L=\boldsymbol{\lambda}f.\hat{f}$. Then for $F\subseteq [K\times K'\rightarrow K'']$ directed 
	\begin{align*}
		L(\bigcurlyvee F)&=\boldsymbol{\lambda} x\boldsymbol{\lambda} y.(\bigcurlyvee F)(x,y) \\
		&=\boldsymbol{\lambda} x\boldsymbol{\lambda} y.\bigcurlyvee_{f\in F}f(x,y) \\
		&\simeq \bigcurlyvee_{f\in F}\boldsymbol{\lambda} x\boldsymbol{\lambda} y.f(x,y) \\
		&= \bigcurlyvee L(F).
	\end{align*}
\end{proof}

\begin{defin}[$CHPO$]
The simplicial category $CHPO_{\Delta}$ is defined as follows:
\begin{enumerate}
    \item The objects of $CHPO_{\Delta}$ are c.h.p.o's.
    \item Given c.h.p.o's $K$ and $K'$, we define 
    $Map_{CHPO_\Delta}(K,K')$ to be the Kan complex $[K\rightarrow K']$ of Definition \ref{The chpo [K rightarrow K']}.
\end{enumerate}
We let CHPO denote the homotopy coherent nerve $\operatorname{N_{\bullet}^{hc}}(CHPO_\Delta)$. We will refer to CHPO as the $\infty$-category of the c.h.p.o's.
\end{defin}

\begin{rem}[\cite{Martinez2HoDTvsHoTT21}]
    Note that subcategory $CHPO\subseteq Cat_\infty$, whose objects are the c.h.p.o's and the morphisms are the continuous functors.
\end{rem}

\begin{prop}[\cite{Tese}]
	$CHPO$ is a Cartesian closed $\infty$-category. 
\end{prop}
\begin{proof}[Proof]
	One has that the product of c.h.p.o's $K\times K'\in CHPO$. The singleton c.h.p.o $\Delta^0$ is a terminal object. By \ref{Proposition-Continuity-application} and \ref{Proposotion-Continuity-abstraction} for each continuous functor  $f:K\times K'\rightarrow K''$ there is a unique continuous functor $\hat{f}:K\rightarrow [K'\rightarrow K'']$ such that
	\[\xymatrix{
		& {K\times K'}\ar@{-->}[dd]_{\hat{f}\times id_{K'}}\ar[rr]^f& & K''
		\\ \\
		& [K'\rightarrow K'']\times K'\ar[rruu]_{Ap}& & 
		& 
	}\]
	commutes in $sSet$ (functors are morphisms of simplicial sets). Thus, the functor $K\times (-)$ has a right adjoint functor $[K\rightarrow (-)]$. 	 
\end{proof}

Another alternative:

\begin{rem}
	Let $g\in[K\rightarrow [K'\rightarrow K'']]$. Define the functor $\check{g}(x,y)=\boldsymbol{\lambda}(x,y)\in K\times K'.g(x)(y)$ . Then 
	\begin{enumerate}
		\item $\check{g}$ is continuous, i.e., $\check{g}\in [K\times K'\rightarrow K'']]$;
		\item $\boldsymbol{\lambda}g.\check{g}:[K\rightarrow [K'\rightarrow K'']]\rightarrow[K\times K'\rightarrow K'']]$ is continuous.
		\item $\boldsymbol{\lambda}g.\check{g}$ is an inverse of $\boldsymbol{\lambda}f.\hat{f}$.  
	\end{enumerate}
	Hence, $CHPO$ is Cartesian closed. 
\end{rem}

Before showing that $CHPO$ admits colimits of the $\omega$-diagrams, we see the following results.

\begin{prop}\label{Colimit diagram in CHPO}
    Let $\mathcal{C}$ be a small category and let $\overline{\mathscr {F}}: \operatorname{\mathcal{C}}^{\triangleright } \rightarrow CHPO_\Delta$ be a diagram of $\infty$-categories. Then $\mathscr {F}$ is a categorical colimit diagram (in the sense of \cite[\href{https://kerodon.net/tag/03CQ}{Definition 03CQ}]{kerodon}) if and only if the induced functor of $\infty$-categories 
    $$\operatorname{N}_{\bullet }^{\operatorname{hc}}( \overline{\mathscr {F}}): \operatorname{N}_{\bullet }( \operatorname{\mathcal{C}}^{\triangleright } ) \rightarrow \operatorname{N}_{\bullet }^{\operatorname{hc}}(CHPO_\Delta) = CHPO$$
    is a colimit diagram (Definition \ref{Definition de Colimit and limit}). 
\end{prop}
\begin{proof}
     Combine Proposition \ref{Proposition 03D2 kerodon} with Proposition  \ref{Corollary 03BB kerodon}.
\end{proof}

\begin{prop}
    The inclusion functor $\iota : \operatorname{N}_{\bullet}(CHPO_\Delta) \hookrightarrow \operatorname{N}_{\bullet}^{\operatorname{hc}}(CHPO_\Delta) = \operatorname{CHPO}$ preserves colimits indexed by small filtered categories. 
\end{prop}
\begin{proof}\label{CHPO admits small colimit filtered}
     Let $\mathcal{C}$ be a small filtered category and suppose that we are given a colimit diagram  $\overline{\mathscr {F}}: \operatorname{\mathcal{C}}^{\triangleright } \rightarrow CHPO_\Delta$ in the ordinary category $CHPO_\Delta$. We wish to show that the induced map $\operatorname{N}_{\bullet }^{\operatorname{hc}}( \overline{\mathscr {F}} ): \operatorname{N}_{\bullet }(\operatorname{\mathcal{C}^{\triangleright }}) \rightarrow CHPO$ is a colimit diagram in the $\infty$-category $CHPO$. By Proposition \ref{Colimit diagram in CHPO}, this is equivalent to the requirement that $\overline{\mathscr {F}}$ is a categorical colimit diagram, which follows from Proposition \ref{Proposition 03DE kerodon}.
\end{proof}

\begin{corol}\label{CHPO admits omega-diagrams}
The $\infty$-category $CHPO$ admits colimits of the $\omega$-diagrams.    
\end{corol}
\begin{proof}
    Given the $\omega$-diagram $( \{K_n \} _{n \geq 0}, \{  f_ n \} _{n \geq 0} )$ in $CHPO$. Let $\underrightarrow{clim}(K_n,f_n)$ denote the colimit of this diagram (formed in the ordinary category of simplicial sets). Then $\underrightarrow{clim}(K_n,f_n)$ is also a c.h.p.o, which is also a colimit of the associated diagram $\operatorname{N}_{\bullet }( \operatorname{\mathbf{Z}}_{\geq 0} ) \rightarrow CHPO$. This is a special case of the Proposition \ref{CHPO admits small colimit filtered}, since $\operatorname{\mathbf{Z}}_{\geq 0}$ (denote the set of non-negative integers, endowed with its usual ordering) is a small filtered category. 
\end{proof}

Next, Definition \ref{Definition of K-infinity} and Remark \ref{K_infty is a chpo}; An analytical view of the limits of $\omega^{op}$-diagrams in $CHPO$ for the dual case of Corollary \ref{CHPO admits omega-diagrams}

\begin{defin}[The Kan complex $K_\infty$ \cite{Tese}]\label{Definition of K-infinity}
	Let $\{K_i\}_{i\in\omega}$ be countable sequence of c.h.p.o's and let $f_i\in [K_{i+1}\rightarrow K_{i}]$ for each $i\in\omega$.
	\begin{enumerate}
		\item The diagram $(K_i,f_i)$ is called a projective (or inverse) system of c.h.p.o's.
		\item The projective  limit of the system $(K_i,f_i)$ (notation $\underleftarrow{lim} (K_i,f_i)$) is the h.p.o $(K_\infty,\precsim_{\infty})$, where $K_\infty$ is a subcategory of  $\,\Pi_{i} K_i$ ($\omega$-times Cartesian product) whose $n$-simplexes (for any $n\geq 0$) are the sequences $(\sigma_i)_{i\in\omega}$ such that $\sigma_i:\Delta^n\rightarrow K_i$, $f_i(\sigma_{i+1})\simeq \sigma_i$ and the order on objects of $K_\infty$ defined by
		$$(x_i)_i\precsim_{\infty} (y_i)_i~~\text{if}~~\forall i.~ x_i\precsim_i y_i~~(\text{in $K_i$}).$$
	\end{enumerate}
\end{defin}

\begin{rem}\label{K_infty is a chpo}
	Let $(K_i,f_i)$ be a projective system. Then  $\underleftarrow{lim} (K_i,f_i)=K_\infty$ is c.h.p.o with
	$$\bigcurlyvee X=\boldsymbol{\lambda} i.\bigcurlyvee\langle x(i)\,|\,x\in X\rangle,$$
	for directed $X\subseteq \underleftarrow{lim} (K_i,f_i)$. Since, if $X$ is directed, then $\langle x(i)\,|\,x\in X\rangle$ is directed for each $i$. Let 
	$$y_i=\bigcurlyvee\langle x(i)\,|\,x\in X\rangle.$$
	Then by the continuity of $f_i$
	\begin{align*}
		f_i(y_{i+1})&\simeq\bigcurlyvee f_i(\langle x(i+1)\,|\,x\in X\rangle) \\
		&=\bigcurlyvee\langle x(i)\,|\,x\in X\rangle \\
		&=y_i
	\end{align*}
	Thus, $(y_i)_i\in\underleftarrow{lim} (K_i,f_i)$. Clearly, it is the supremum of $X$.

\end{rem}

\begin{prop}[\cite{Tese}]
    The $\infty$-category $CHPO$ is a 0-$\infty$-category.
\end{prop}
\begin{proof}
    \begin{enumerate}
        \item Every Kan complex $CHPO(K,K')=[K\rightarrow K']$ is a chpo according to Proposition \ref{Proposition-Continuous-Funtor-chpo}, where the constant functor $\lambda x.0'$ is a minimum element of $CHPO(K,K')$
        \item Let $g\in [K'\rightarrow K'']$ and a directed $F\subseteq [K\rightarrow K']$. Then, the composition of $f$ with the supremun of $F$, is given by
        $(g.(\bigcurlyvee F))(x)=g(\bigcurlyvee\langle f(x)\,|\,f\in F\rangle)$; by the definition of $\bigcurlyvee F$ in Proposition \ref{Proposition-Continuous-Funtor-chpo}. By the continuity of $g$, we have $(g.(\bigcurlyvee F))(x)\simeq \bigcurlyvee\langle g(f(x))\,|\,f\in F\rangle:=(\bigcurlyvee g.F)(x)$. The proof of the composition on the left (continuity of $.g$) is similar; used the continuity of the functor $\bigcurlyvee F' \in [K''\rightarrow K''']$.
        \item Let $f\in CHPO(K,K')$ and $0_{K',K''}=\lambda x'. 0''\in CHPO(K',K'')$, the composition $0_{K',K''}.f=(\lambda x'. 0'').f\simeq \lambda x. 0''=0_{K,K''}$.
    \end{enumerate}
\end{proof}

Next, by Remark \ref{Remark-Solution-HDE}, we argue that $K_\infty$ is a solution to the homotopy domain equation $[X \rightarrow X]\simeq X$ in the CHPO category. 
In Section \ref{Section: Contruction of K_infty}, we present further details on the projection pairs used to calculate $K_\infty$.

\begin{rem}\label{Example on CHPO}
		The $\infty$-category $CHPO$  is a 0-$\infty$-category. $CHPO$ is an i.c.c and has limits for every $\omega^{op}$-diagram. The functor $\Rightarrow:CHPO\times CHPO\rightarrow CHPO$ is defined by:
		$$\Rightarrow (A,B)=[A\rightarrow B],\hspace{1cm} \Rightarrow (f,g)=\lambda h.(g.h. f),$$
		is locally continuous, where $g.h.f$ corresponds to composable morphisms in $CHPO$.
		\medskip The diagonal functor $\Delta:CHPO\rightarrow CHPO\times CHPO$, defined by $\Delta(A)=(A,A)$ and $\Delta(f)=(f,f)$ is also locally continuous. Thus, by Theorem \ref{continuous-functor} and conclude that the associated functors
		\begin{enumerate}
			\item $(\Rightarrow)^{+-}:(CHPO)^{HPrj}\times (CHPO)^{HPrj}\rightarrow (CHPO)^{HPrj}$ 
			
			$(\Rightarrow)^{+-}((f^+,f^-),(g^+,g^-))=(\lambda h.(g^+. h. f^-), \,\lambda h.(g^-. h. f^+))$
   
			\item $(\Delta)^{+-}:(CHPO)^{HPrj}
			\rightarrow (CHPO)^{HPrj}\times (CHPO)^{HPrj}$
			
			$(\Delta)^{+-}(f^+,f^-)=((f^+,f^-),(f^+,f^-))$
		\end{enumerate}
	are $\omega$-continuous. However, the composition of $\omega$-continuous functors is still an $\omega$-continuous functor; thus, the functor $F=(\Rightarrow)^{+-}\circ (\Delta)^{+-}:(CHPO)^{HPrj}\rightarrow(CHPO)^{HPrj}$ is $\omega$-continuous. Explicitly, $F$ is defined by
	\begin{align*}
	& F(X)=[X\rightarrow X] \\
	& F(f^+,f^-)=(\lambda h.(f^+. h. f^-),\, \lambda h.(f^-. h. f^+))
	\end{align*}
    if all the $f_i^{-}\in [K_{i+1}\rightarrow K_i]$, of the inverse diagram $(K_i, f_i^-)$, are h-projections (Definition \ref{h-projection}) for some h-embedding $f_i^{+}\in [K_{i}\rightarrow K_{i+1}]$, and $K_\infty =\underleftarrow{lim}(K_i,f_i^-)$, by Theorem \ref{colimit-HPrj-category}, we see that $K_\infty$ is a colimit for the diagram $(K_i, (f_i^+,f_i^-))$ in $(CHPO)^{HPrj}$ (see Section \ref{Section: Contruction of K_infty}, for more details on the initial  h-projection pair $(f_0^+,f_0^-)$ and the initial c.h.p.o $K_0$). Thus, by Theorem \ref{point-fixed-theorem}, we conclude that $K_\infty \simeq F(K_\infty)=[K_\infty\rightarrow K_\infty]$ in $(CHPO)^{HPrj}$.
	\end{rem}

\subsection{Scott Domains}    

\begin{defin}
Let $K\subseteq\hat{K}$ be a c.h.p.o.
	\begin{enumerate}
		\item $x\in K$ is compact if for every directed $X\subseteq K$ one has
		$$x\precsim\bigcurlyvee X~~\Longrightarrow~~x\precsim x_0~\text{for some $x_0\in X$}.$$ 
		\item $K$ is an algebraic c.h.p.o if for all $x\in K$ the h.p.o $x\downarrow= \langle y\precsim x\,|\,y~\text{compact}\rangle$ is directed and $x\simeq\bigcurlyvee (x\downarrow).$
        \item $K$ is bounded complete if for every h.p.o $X\subseteq K$ with an upper bound has supremum. 
	\end{enumerate}
\end{defin}

\begin{defin}
    A Homotopy Scott Domain is a bounded complete algebraic c.h.p.o.
\end{defin}

\begin{prop}[\cite{Tese}]\label{continuity by compacts}
	Let $K$ be algebraic and $f:K\rightarrow K$. Then $f$ is continuous iff $f(x)\simeq\bigcurlyvee\langle f(e)\,|\,e\precsim x~\text{and}~e~\text{compact}\rangle$.
\end{prop}
\begin{proof}[Proof]
	$(\Rightarrow)$ Let $f$ be continuous. Then 
	\begin{align*}
		f(x)&=f(\bigcurlyvee\langle e\precsim x\,|\,e~\text{compact}\rangle) \\
		&\simeq\bigcurlyvee\langle f(e)\,|\,e\precsim x~\text{and}~e~\text{compact}\rangle. 
	\end{align*}
	$(\Leftarrow)$ First we check that $f$ is monotonic. If $x\precsim y$, then
	$$\langle e\precsim x\,|\,e~\text{compact}\rangle\subseteq\langle e\precsim y\,|\,e~\text{compact}\rangle,$$
	hence
	\begin{align*}
		f(x)&\simeq\bigcurlyvee\langle f(e)\,|\,e\precsim x~\text{and}~e~\text{compact}\rangle \\
		&\precsim \bigcurlyvee\langle f(e)\,|\,e\precsim y~\text{and}~e~\text{compact}\rangle \\
		&\simeq f(y).
	\end{align*}
	Now let $X\subseteq K$ directed. Then
	\begin{align*}
		f(\bigcurlyvee X)&\simeq\bigcurlyvee\langle f(e)\,|\,e\precsim \bigcurlyvee X~\text{and}~e~\text{compact}\rangle \\
		&\precsim\bigcurlyvee\langle f(x)\,|\,x\in X\rangle;\hspace{0.5cm}\text{by compactness and monotonicity of $f$} \\
		&\precsim f(\bigcurlyvee X);\hspace{0.5cm}\text{by monotonicity and definition of supremum}.
	\end{align*}
	Thus, $f(\bigcurlyvee X)\simeq\bigcurlyvee f(X)$.
\end{proof}
\begin{prop}[\cite{Tese}]
	Let $K,K'$ be c.h.p.o's.
	\begin{enumerate}
		\item $(x,y)\in K\times K'$ is compact iff $x$ and $y$ are compact.
		\item $K$ and $K'$ are algebraic, then $K\times K'$ is algebraic.
        \item $K$ and $K'$ are bounded complete, then $K\times K'$ is bounded complete.
	\end{enumerate}
\end{prop}
\begin{proof}[Proof]
	1. $(\Rightarrow)$ Let $X\subseteq K$ and $Y\subseteq K'$ be directed such that $x\precsim\bigcurlyvee X$ and $y\precsim\bigcurlyvee Y$. Hence $$(x,y)\precsim(\bigcurlyvee X,\bigcurlyvee Y)=\bigcurlyvee (X\times Y).$$
	By hypothesis, $(x,y)\precsim (x_0,y_0)$ for some $(x_0,y_0)\in X\times Y$. Thus, $x$ and $y$ are compact.
	
	\bigskip $(\Leftarrow)$ Let $X\subseteq K\times K'$ be directed such that $$(x,y)\precsim\bigcurlyvee X= (\bigcurlyvee X_0, \bigcurlyvee X_1)=\bigcurlyvee (X_0\times X_1).$$  
	By hypothesis, there are $x_0\in X_0$ and $y_0\in X_1$  such that $x\precsim x_0$ and $y\precsim y_0$. Since $(x_0,y_0)\precsim \bigcurlyvee (X_0\times X_1)=\bigcurlyvee X$, there exists $(x_1,y_1)\in X$ such that $(x,y)\precsim (x_0,y_0)\precsim (x_1,y_1)$. Hence $(x,y)$ is compact.
	
	\bigskip 2. Let $(x,y)\in K\times K'$. Then
	\begin{align*}
		\bigcurlyvee((x,y)\downarrow)&=\bigcurlyvee\langle(e,d)\precsim (x,y)\,|\,(e,d)~\text{compact}\rangle \\
		&=\bigcurlyvee\langle(e,d)\precsim (x,y)\,|\,\text{$e$ and $d$ are compact}\rangle;\hspace{0.5cm}\text{by 1.}\\
		&\simeq (	\bigcurlyvee (x\downarrow),	\bigcurlyvee (y\downarrow)) \\
		&\simeq (x,y);\hspace{0.5cm}\text{by hypothesis.}
	\end{align*}

    \bigskip 3. Let $K$ and $K'$ be bounded complete c.h.p.o's and the h.p.o $X\subseteq K\times K'$ with $(x_0,x_1)$ be an upper bound. By hypothesis we have
    $$\bigcurlyvee X= (\bigcurlyvee X_0, \bigcurlyvee X_1)\precsim (x_0,x_1)$$ 
\end{proof}

\begin{defin}[$Alg$ \cite{Tese}]
	Define the full subcategory $Alg\subseteq CHPO$, whose objects are the algebraic c.h.p.o's and the morphisms are continuous functors. 	
\end{defin}

Note that $Acc$, the $\infty$-category of accessible $\infty$-categories, is the generalisation of $Alg$, since all directed is filtered, every supremum is a filtered colimit, and for each $X\in Alg$, $X$ has all supremum and the full subcategory $X_c\subseteq X$ of compact objects generates $X$ under supremum.

\begin{defin}
    	Define the full subcategory $HD\subseteq CHPO$, whose objects are Homotopy Scott Domains and the morphisms are continuous functors.
\end{defin}

\begin{prop}\label{K_n is Scott Domain}
    If $K_0$ is a Homotopy Scott Domain, then $K_{n+1}=[K_n\rightarrow K_n]$ is a Homotopy Domain Scott for each $n\geq 0$. 
    \end{prop}
    \begin{proof}
We suppose that $K_n$ is algebraic (IH) and will see that by induction $K_{n+1}=[K_n\rightarrow K_n]$ is also algebraic. The compact elements of $K_{n+1}$ are functors of the form $\bigcurlyvee_{i=1} ^k(a_i\rightarrow b_i)$, where $a_i,b_i\in K_n$ are compacts, and $(a\Rightarrow b):K\rightarrow K$ denote the step functor induced by the map
 \[ (a\Rightarrow b)(x):= \begin{cases}
  b & \mbox{ if $ x\succsim a $}\\
  0 & \mbox{o.c}
  \end{cases}
  \]

  \noindent by degenerate higher simplexes; for example, for the case of an equivalence $x\xrightarrow{\sim} y$ at $K$, we have the induced map
   \[ (a\Rightarrow b)(x\xrightarrow{\sim} y):= \begin{cases}
    1_b & \mbox{ if $ x \succsim a $}\\
    1_0 & \mbox{o.c.}
  \end{cases}
  \]

  In the same way, the other maps are induced for degenerate images of higher simplexes $\sigma:\Delta^n\rightarrow K$. These functors are compact because they depend only on a finite number of values of $K_n$. We will see how to approximate functors using compacts: Let $f\in K_{n+1}$. For each $x\in K_n$, since $K_n$ is algebraic (by IH), we have $x\simeq\bigcurlyvee(x\downarrow)=\bigcurlyvee \langle a\,|\,a\precsim x~\text{and}~a~\text{compact}\rangle$. By the continuity of $f$ and Proposition \ref{continuity by compacts}, $f(x)\simeq\bigcurlyvee\langle f(a)\,|\,a\precsim x~\text{and}~a~\text{compact}\rangle$. Since $f(a)\in K_n$, and $K_n$ is algebraic, $f$ is the supremum of compacts: $f(a)\simeq \bigcurlyvee\langle b\,|\,b\precsim f(a)~\text{and}~b~\text{compact}\rangle$. Thus, the hpo of compact functors below of $f$ is given by: 
  $$(f\downarrow)=\langle \bigcurlyvee_{i=1}^k (a_i\Rightarrow b_i):b_i\precsim f(a_i)~\text{and}~  a_i,b_i\in K_n \,\text{compacts}\rangle.$$
  Since $(f\downarrow)$ is directed (admits finite colimits) and $f(x)$ is approximated by values $b\precsim f(a)$ for $a\precsim x$, we have $f(x)\simeq \bigcurlyvee \langle g(x):g\in (f\downarrow)\rangle$.

\medskip We suppose that $K_n$ is a bounded complete (IH) and will see that by induction $K_{n+1}=[K_n\rightarrow K_n]$ is also bounded complete. Let hpo $F\subseteq [K_n\rightarrow K_n]$ with an upper bound $h\in [K_n\rightarrow K_n]$. For each $x\in K_n$, we define $(\bigcurlyvee F)(x):=\bigcurlyvee \langle f(x):f\in F\rangle$; It is well defined, since $h$ is an upper bound of $F$, $f(x)\precsim h(x)$ for each $f\in F$, given that $K_n$ is bounded complete, by IH, there exists the supremum $\bigcurlyvee \langle f(x):f\in F\rangle$. The proof that $\bigcurlyvee F$ is continuous is similar to the proof of Proposition \ref{Proposition-Continuous-Funtor-chpo} and Lemma \ref{lemma-complete}. 
\end{proof}

\section{Construction of $K_\infty$ by h-projection pairs}\label{Section: Contruction of K_infty}

\begin{defin}[\cite{Martinez21}]
	Let $X$ be a Kan complex. Define $\pi_0(X):=\pi_0(|X|)$ and $\pi_n(X,x):=\pi_n(|X|,x)$ for $n>0$, with $|\,\,\,|:Kan\rightarrow Top$ being the functor of geometric realisation from the category of Kan complexes to the category of the topological spaces.
\end{defin}

The fact that a Kan complex $X$ is not contractible does not imply that every vertex $x\in X$ contains relevant information, that is, the higher fundamental groups $\pi_n(|X|,x)$ are non-trivial, nor that it contains holes in all the higher dimensions. To guarantee the existence of non-trivial Kan complexes as higher $\lambda$-models, we present the following definition.

\begin{defin}[Non-trivial Kan complex \cite{Martinez21}]
	A small Kan complex $X$ is non-trivial if
	\begin{enumerate}
		\item  $\pi_0(X)$ is infinite. 
		\item for each $n\geq 1$, there is a vertex $x\in X$ such that $\pi_n(X,x)\ncong\ast$.
		\item for each vertex $x$ of some non-degenerated horn in $X$, there is $n\geq 1$ such $\pi_n(X,x)\ncong\ast$.
	\end{enumerate}
	
\end{defin}

\begin{ejem}\label{Example1}
	For each $n\geq 0$, let $S^n$ the sphere of Remark \ref{Remark}. Define the non-trivial Kan complex $N$ as the disjoint union:
	$$N=\coprod_{n<\omega}S^n.$$
\end{ejem}

 Furthermore, $\pi_{n}(S^n)\ncong\ast$ for all $n\geq 1$, and there is $k > n$ such that $\pi_{k}(S^n)\ncong\ast$ for each $n\geq 2$ \cite{DBLP:books/mk/Hatcher01}. 
 
\begin{defin}[The c.h.p.o $N^+$]
    Define the c.h.p.o $N^+:=\langle N\cup \{\bot_0\}\rangle$, where $\bot_0$ is the unique initial object of $N^+$.
\end{defin}

\begin{defin}[Non-Trivial Homotopy $\lambda$-Model \cite{Martinez21}]
	Let $\mathcal{K}$ be a Cartesian closed $\infty$-category, whose objects are Kan complexes. An object $K\in\mathcal{K}$ is a non-trivial homotopy $\lambda$-model, if $K$ is a non-trivial Kan complex and it is reflexive. 
\end{defin}

\begin{defin}[Sequence $K_0$, $K_1\ldots$]
    For each natural $n\geq 0$ define $K_n$ by recursion
 \begin{align*}
&K_0 := N^+, 
\\
&K_{n+1}:= F(K_n)=[K_n\rightarrow K_n].
\end{align*}
Where $F:(CHPO)^{HPrj}\rightarrow(CHPO)^{HPrj}$ is the functor defined in the Remark \ref{Example on CHPO} and $[K_n\rightarrow K_n]$ be the c.h.p.o of the $\omega$-continuous functors or continuous. Note by $\bot_n$ the unique initial object of $K_n$.
\end{defin}

\begin{rem}\label{Remark K_n is Scott Domain}
    Notice that $K_0=N^+$ is a Scott Domain. By Proposition \ref{K_n is Scott Domain},  $K_{n+1}=[K_n\rightarrow K_n]$ is also a Scott Domain for each $n\geq 0$.
\end{rem}

\begin{defin}[Initial h-projection pair]\label{Definition h-projection initial}
   Define the initial h-projection pair $(f^+_0,f^-_0)$, from $K_0$ to $K_1$ the following way:
   \begin{enumerate}
       \item[(a)] the h-embedding $f_0^+$  from $K_0$ into $K_1$ defined as $f_0^+(x) := \lambda y\in K.x'$,  for some edge $x\rightarrow x'$ at $K_0$.
       \item[(b)] and $f_0^-$ is the unique (under homotopy) h-projection associated with the h-embedding $f_0^+$  such that $f_0^-(g):= g(\bot_0)$.
       \end{enumerate}
\end{defin}

\begin{defin}[h-projection pair from $K_n$ to $K_{n+1}$]
 From the definition of the functor $F(f^+,f^-)=(\lambda h.(f^+. h. f^-), \,\lambda h.(f^-. h. f^+))$  of Remark \ref{Example on CHPO}, we define the h-projection pair $(f^+_n, f^-_n)$. If $n=0$ define $(f^+_0, f^-_0)$ as in Definition \ref{Definition h-projection initial} by recursion. If $n\geq 1$ define the h-projection pair $(f^+_n, f^-_n)$ from $K_n$ to $K_{n+1}$ as
 $$(f_n^+,f_n^-):=F^n(f^+_0,f^-_0)=(\lambda h.(f_{n-1}^+. h. f_{n-1}^-), \,\lambda h.(f_{n-1}^-. h. f_{n-1}^+))$$ 
\end{defin}

\begin{rem}[The $K_\infty$ homotopy $\lambda$-model]
   Let $(K_n,f^-_n)$ be an inverse system of h-projections such that $\underleftarrow{lim}(K_n,f^-_n)=(K_\infty,(f_{\infty,n})_n)$. By the Theorem \ref{colimit-HPrj-category}, we have
    $(K_\infty ,(f_{n,\infty}, f_{\infty, n})_n)$ 
    as a colimit to the diagram $(F^n(K_0), F^n(f_0^+,f_0^-))=(K_n, (f_n^+,f_n^-))$ at the $\infty$-category $(CHPO)^{HPrj}$, where $K_\infty$ corresponds to c.h.p.o defined in \ref{Definition of K-infinity} and the h-embedding $f_{n,\infty}(x):=(f_{n,m}(x))_m$ (from $K_n$ to $K_\infty$) with $f_{n,m}$ equivalent to the composable h-embeddings if $n<m$, or to the composable h-projections if $n>m$, or to the identity if $n=m$, and the h-projection  $f_{\infty, n}(x):=x_n$ (from $K_\infty$ to $K_n$).

    \medskip Note that by Remark \ref{Example on CHPO}, the c.h.p.o $K_\infty$ is a fixed point of the functor $F: (CHPO)^{HPrj}\rightarrow (CHPO)^{HPrj}$, that is, $K_\infty\simeq [K_\infty\rightarrow K_\infty]$. Therefore, $K_\infty$ is a homotopy $\lambda$-model.
\end{rem}

\begin{prop}
    The $K_\infty$ chpo is a Homotopy Scott Domain. 
\end{prop}
    \begin{proof}
         First, we will find that $K_\infty$ is algebraic. The compact elements of $K_\infty$ are of the form $(x_1,\ldots,x_n,\ldots,f^{+}_{n,m}(x_n),\ldots)$ for some $n\in\omega$ and for all $m>n$. For any $x\in K_\infty$, we define its finite approximations as compact elements $x^{(n)}:=(x_1,\ldots,x_n,\ldots,f^{+}_{n,m}(x_n),\ldots)$. The chain $\langle x^{(n)}\rangle_{n\in\omega}$ is directed and $\bigcurlyvee_nx^{(n)}\simeq x$, thus $K_\infty$ is an algebraic chpo. Now we will see that $K_\infty$ is bounded complete. Let $X\subseteq K_\infty$ with an upper bound $u$, then for each $n\in\omega$, the hpo $\langle x_n: x\in K_\infty \rangle\subseteq\hat{K}_n$ is an upper bound by $u_n\in K_n$. Since $K_n$ is bounded complete ($K_n$ is a Scott Domain by Remark \ref{Remark K_n is Scott Domain}), there exists supremum $s_n=\bigcurlyvee\langle x_n: x\in K_\infty \rangle$. By the continuity of $f^{-}_{n-1}$, the sequence $(s_n)_n\in K_\infty$ and is the supremun of $X$, therefore $K_\infty$ is a Homotopy Scott Domain.      
    \end{proof}

\begin{prop}\label{Proposition supremum at K_infty}
    If $x\in K_\infty$, then $x \simeq \bigcurlyvee_n f_{n,\infty}(x_{n})$.
\end{prop}
\begin{proof}
We wish to show that for every $k\in\omega$, $(\bigcurlyvee_n f_{n,\infty}(x_{n}))_k\simeq x_k$. By the definition and continuity of $f_{\infty,k}$, we have:
\begin{align*}
(\bigcurlyvee_n f_{n,\infty}(x_{n}))_k &= f_{\infty,k}(\bigcurlyvee_n f_{n,\infty}(x_{n}))\\
    &\simeq \bigcurlyvee_n f_{\infty,k}(f_{n,\infty}(x_{n})) \\
    &\simeq \bigcurlyvee_{n\geq k}f_{n,k}(x_{n}) \\
    &\simeq \bigcurlyvee_{n\geq k} x_{k}  \\
    &\simeq x_k
\end{align*} 
\end{proof}

\begin{prop}\label{Proposition equivalences h, k}
    If $h: [K_\infty\rightarrow K_\infty]\rightarrow K_\infty$ and $k: K_\infty\rightarrow [K_\infty\rightarrow K_\infty]$ are the equivalences resulting from applying Theorem \ref{point-fixed-theorem} to the functor $F: (CHPO)^{HPrj}\rightarrow (CHPO)^{HPrj}$ of Remark \ref{Example on CHPO}, then
    $$h(f)\simeq \bigcurlyvee_n f_{n+1,\infty}(f_{\infty, n}.f.f_{n, \infty}),\hspace{0.5cm} k(x)(y)\simeq\bigcurlyvee_n f_{n,\infty}(x_{n+1}(y_{n})).$$
\end{prop}
\begin{proof}
    We prove $(\lambda f.\bigcurlyvee_k f_{k+1,\infty}(f_{\infty, k}.f.f_{k, \infty})).F(f_{n,\infty})\simeq f_{n+1,\infty}$. We have $F(f_{n,\infty})=\lambda g.(f_{n,\infty}.g.f_{\infty, n}):K_{n+1}\rightarrow F(K_\infty)$, and hence
     \begin{align*}
&(\lambda f.\bigcurlyvee_k f_{k+1,\infty}(f_{\infty, k}.f.f_{k, \infty}))(F(f_{n,\infty})(g)) 
\\
&\simeq\bigcurlyvee_k f_{k+1,\infty}(f_{\infty, k}.F(f_{n,\infty})(g).f_{k, \infty})
\\
&\simeq \bigcurlyvee_{k>n} f_{k+1,\infty}(f_{\infty, k}.f_{n,\infty}.g.f_{\infty, n}.f_{k, \infty}) 
\\
&\simeq \bigcurlyvee_{k>n} f_{k+1,\infty}(f_{n,k}.g.f_{k, n}) 
\\
&\simeq \bigcurlyvee_{k>n} f_{k+1,\infty}(f_{n+1,k+1}(g))
\\
&\simeq \bigcurlyvee_{k>n} f_{n+1,\infty}(g)
\\
&\simeq f_{n+1,\infty}(g)
\end{align*}
 Since  $(K_\infty ,(f_{n,\infty}, f_{\infty, n})_n)$ is a colimit  to the $\omega$-diagram $(K_n,(f_n^+,f_n^-))$, there exists a unique  (under homotopy; the space of choices is contractible) $h: [K_\infty\rightarrow K_\infty]\rightarrow K_\infty$ such that $h.F(f_{n,\infty})\simeq f_{n+1,\infty}$, therefore $h(f)\simeq \bigcurlyvee_n f_{n+1,\infty}(f_{\infty, n}.f.f_{n, \infty})$.
 
 \bigskip On the other hand, by the same colimit, there exists a unique (under homotopy) $k:K_\infty\rightarrow [K_\infty\rightarrow K_\infty]$ such that $k.f_{n+1,\infty}\simeq F(f_{n,\infty})$, applying $f_{\infty, n+1}$ on both sides
\begin{align*}
 k.f_{n+1,\infty}.f_{\infty, n+1}&\simeq F(f_{n,\infty}).f_{\infty, n+1}
\\
& \simeq (\lambda g.(f_{n,\infty}.g.f_{\infty,n})).f_{\infty, n+1}; \,\, \text{by definition of} \,\, F(f_{n,\infty})
\\
 k(f_{n+1,\infty}(f_{\infty, n+1}(x))&\simeq(\lambda g.(f_{n,\infty}.g.f_{\infty,n}))(f_{\infty, n+1}(x))
\\
&\simeq f_{n,\infty}.f_{\infty, n+1}(x).f_{\infty,n}
\\
k(f_{n+1,\infty}(x_{n+1}))(y)&\simeq (f_{n,\infty}.x_{n+1}.f_{\infty,n})(y);\,\, \text{by definition of} \,\, f_{\infty, n+1}(x)
\\
&\simeq f_{n,\infty}(x_{n+1}(f_{\infty,n}(y)))
\\
&\simeq f_{n,\infty}(x_{n+1}(y_n));\,\, \text{by definition of} \,\, f_{\infty, n+1}(y)
\end{align*}
\begin{align*}
\bigcurlyvee_n k(f_{n+1,\infty}(x_{n+1}))(y)&\simeq \bigcurlyvee_n f_{n,\infty}(x_{n+1}(y_n))
\\
k(\bigcurlyvee_n f_{n+1,\infty}(x_{n+1}))(y)&\simeq \bigcurlyvee_n f_{n,\infty}(x_{n+1}(y_n))
\\
k(x)(y)&\simeq \bigcurlyvee_n f_{n,\infty}(x_{n+1}(y_n)); \,\, \text{since} \,\,  x \simeq \bigcurlyvee_n f_{n+1,\infty} (x_{n+1}) \,\ \text{by} \,\, \ref{Proposition supremum at K_infty}.
\end{align*}
\end{proof}

\begin{rem}[Application in $K_\infty$]
    For $x,y \in K_\infty$, by definition $x\bullet y = k(x)(y)$, by Proposition \ref{Proposition equivalences h, k} one has $$x\bullet y \simeq \bigcurlyvee_n f_{n,\infty}(x_{n+1}(y_n)).$$
\end{rem}

\begin{prop}
    The Kan complex $K_\infty$ is a non-trivial homotopy $\lambda$-model.
\end{prop}
\begin{proof}
    \begin{enumerate}
        \item  Since $\pi_0(K_0)$ is infinite and $f_{0,\infty}$ is an h-embedding, the map of connected components  $\pi_0(f_{0,\infty}):\pi_0(K_0)\rightarrow\pi_0(K_\infty)$ is injective, and hence $\pi_0(K_\infty)$ is infinite.
        \item Given $n\geq 1$, one has $\pi_n(K_0,x_0)=\pi_n (S^n,x_0)\ncong\ast$. By the h-embedding $f_{0,\infty}$, we have that the group homomorphism $\pi_n(f_{0,\infty}):\pi_n
        (K_0,x_0)\rightarrow\pi_n(K_\infty,f_{0,\infty}(x_0))$ is injective, thus $\pi_n (K_\infty,f_{0,\infty}(x_0))\ncong\ast$. 
        \item We find that $K_0=N^+$ is a non-trivial Kan complex. Let $x=(x_i)_i$ be a vertex of some non-degenerated horn in $K_\infty$. To prove that there exists an \( n>0 \) such that $\pi_n(K_\infty, x)\ncong\ast$, it is enough to prove that there is an \( i\geq 0\) such that $\pi_n(K_i, x_i)\ncong\ast$ for some $n>0$. Since $x_0$ is a vertex of a non-degenerated horn in $K_0$, there exists $n>0$ such that $x_0\in S^n$, hence $\pi_n(K_0,x_0)=\pi_n (S^n,x_0)\ncong\ast$.
    \end{enumerate}
\end{proof}

\begin{ejem}[Interpretation of $\beta\eta$-contractions in $K_\infty$]\label{Example beta eta contractions in K_infty}
    Given $a,b\in K_\infty$ such that $a=f_{0,\infty}(0)$ and $b=f_{0,\infty}(1)$. Let  $f=f_{0,\infty}(\langle 1,0\rangle')\in K_\infty(b,a)$ and  $g=f_{0,\infty}(\langle 0,1\rangle)\in K_\infty(a,b)$, where $\langle 1,0\rangle'$ and $ \langle 0,1\rangle$ are not inverses, and $0,1$ are vertices of the degenerated sphere $S^1\subseteq K_0$:
   \[\xymatrix{
		& 1\ar[d]_{\langle1,0\rangle'}\ar[dr]^{i}  \\
		&0 \ar[r]_{\langle 0,1\rangle} & 1 
	}\]
    If we define: $f_{0,1}(0):=f_{0,1}(1):=\lambda x. 1$ (note that $f_{0,n+1}(0)=f_{0,n+1}(1)=1_{n+1}:=\lambda x\in K_n.1_n$, with $1_0:=1$, and $f_{n+1,0}(1_{n+1})=1$), $h(f):=\bigcurlyvee_n f_{n+1,\infty}(f_{\infty, n}.f.f_{n, \infty})$ and $k(x)(y):=\bigcurlyvee_n f_{n,\infty}(x_{n+1}(y_{n}))$, we have
    \begin{align*}
            h(k(a))=&h(k(f_{0,\infty}(0))) \\
            = &\bigcurlyvee_n f_{n+1,\infty}(f_{\infty,n}.f_{n,\infty}.f_{0,n+1}(0).f_{\infty,n}.f_{n,\infty}) \\
            = &\bigcurlyvee_n f_{n+1,\infty}(f_{\infty,n}.f_{n,\infty}.f_{0,n+1}(0));\\
            &(\text{since: }f_{0,n+1} (0).f_{\infty,n}.f_{n,\infty}=1_{n+1}.f_{\infty,n}.f_{n,\infty}=1_{n+1})\\ = &\bigcurlyvee_n f_{n+1,\infty}(f_{\infty,n}.f_{n,\infty}.1_{n+1})\\
            =&\bigcurlyvee_n f_{n+1,\infty}(1_{n+1});\hspace{0.3cm} f_{m,n}(1_m)=1_n,\hspace{0.3cm}\text{for every $m,n\in \omega$} \\
            =&\bigcurlyvee_n f_{n+1,\infty}(f_{0,n+1}(1)) \\
            =&\bigcurlyvee_n f_{0,\infty}(1) \\
            =&f_{0,\infty}(1)=(1,1_1,1_2,\ldots,1_n,\ldots)=b \\
            \xleftarrow[\eta_a]{\neq} &f_{0,\infty}(0)=(0,1_1,1_2,\ldots,1_n,\ldots)=a.
     \end{align*}

    The $\lambda$-term $(\lambda z.xz)y$ has the $\beta\eta$-contractions 
	\[\xymatrix{
		& {(\lambda z.xz)y}\ar@/^2mm/[r]^{\,\,\,1\beta} \ar@/_2mm/[r]_{\,\,\,1\eta} & xy
	}\]
	
	Let $\rho(x)=a$ and $\rho(y)=b$ be vertices of $K_\infty$. The interpretation of a $\lambda$-term is given by: $\llbracket(\lambda z.xz)y \rrbracket_\rho=\llbracket\lambda z.xz \rrbracket_\rho\bullet b=h(\boldsymbol{\lambda} c.k(a)(c))\bullet b=h(k(a))\bullet b=k(h(k(a)))(b)$. Since $h(k(a))=b$, so the interpretation of the $\beta$-contraction in $K_\infty$ is given by
    $$b\bullet b=k(b)(b)=k(h(k(a)))(b)=(kh)(k(a))(b)\xrightarrow[]{\langle k(b)(b),k(a)(b)\rangle'} k(a)(b)=a\bullet b.$$
Thus,  $\varepsilon_{k(a)}:=\langle k(b),k(a)\rangle' = k(\langle b,a\rangle')=k(f_{0,\infty}\langle 1,0\rangle')$. On the other hand, the (reverse) $\eta$-contraction is modelled by 
$$a\bullet b=k(a)(b)\xrightarrow{\langle k(a)(b),k(b)(b)\rangle} k((hk)(a))(b)=k(h(k(a)))(b)=k(b)(b)=b\bullet b$$
Let  $k(\eta_{a}):=\langle k(a),k(b)\rangle = k(\langle a,b\rangle)=k(f_{0,\infty}\langle 0,1\rangle)$. Then, the degenerated sphere $k(f_{0,\infty}(S^1))(b)\subseteq K_\infty$ is given by
    \[\xymatrix{
		& (kh)(k(a))(b)\ar[d]_{(\varepsilon_{k(a)})_b}\ar[dr]^{i}  \\
		&k(a)(b)\ar[r]_{(k(\eta_a))_b} & k((hk)(a))(b) 
	}\]
Therefore,  $\llbracket 1\beta \rrbracket_\rho \ncong \llbracket 1\eta \rrbracket_\rho$ (non-equivalents) in $K_\infty$.  
\end{ejem} 

Example \ref{Example beta eta contractions in K_infty},  indicates that there is no $2\beta\eta$-conversion from $1\beta$ to $1\eta$ in the Homotopy Type-Free Theory (HoTFT) defined in \cite{Martinez2HoDTvsHoTT21}. This shows that the higher $\beta\eta$-conversions system of HoTFT is not trivial.

\subsection{Comparison with Other Intensional Models.}

From a broader semantic perspective, the homotopy $\lambda$-model $K_\infty$ can be viewed as complementary to several well-established intensional approaches to the $\lambda$-calculus. In PER-based models (in the sense of Scott, Hyland, and later Streicher and Longley), intensionality arises from the fine structure of partial equivalence relations over an underlying domain, where programs are identified up to observational equivalence induced by realizers. Game semantics (as developed by Abramsky-Jagadeesan-Malacaria and Hyland-Ong) captures intensional distinctions through interaction patterns between programs
and their environments, emphasizing strategies and their dynamic behavior. Realizability models (in the tradition of Kleene and later realizability interpretations of typed and untyped $\lambda$-calculi) focus on computational witnesses and proof relevance, distinguishing terms by the computational content of their realizers. The model $K_\infty$ departs from these approaches by organizing intensional distinctions along higher homotopical dimensions: terms are not identified merely by equality or equivalence, but by structured spaces of paths, homotopies, and higher coherences. In this sense, $K_\infty$ does not aim to subsume PERs, game semantics, or realizability, but rather to offer an orthogonal semantic perspective in which intensionality is expressed geometrically. This homotopical enrichment makes explicit higher identifications that are implicit or collapsed in traditional models, and suggests a semantic setting in which programs, proofs, and their
equivalences can be analyzed uniformly across multiple levels of abstraction.

\section{Conclusions and further work}

We generalise Dana Scott's $\lambda$-model  $D_\infty$ to a non-trivial homotopy $\lambda$-model $K_\infty$, corresponding to a complete weakly ordered Kan complex or complete homotopy partial order (chpo), with holes in higher dimensions to guarantee non-equivalence between interpretations of some higher conversions. For future work, we are developing an alternative type theory interpreted by chpo's (or Scott domains), where the higher $\beta\eta$-conversions are included, in their typed version, as witnesses of the iterated identity types.

\bibliography{mybibfile}

\end{document}